\documentclass[11pt,reqno]{amsart}

\usepackage[a4paper,margin=3cm]{geometry}
\usepackage{amsmath,amssymb,amsthm,amsfonts}
\usepackage{mathrsfs}
\usepackage{hyperref}
\usepackage{enumitem}
\usepackage{bbm}
\usepackage[pagewise]{lineno}
\newtheorem{Theorem}{Theorem}[section]

\newtheorem{Proposition}[Theorem]{Proposition}

\newtheorem{Definition}[Theorem]{Definition}
\newtheorem{rem}[Theorem]{Remark}
\newtheorem{example}[Theorem]{Example}

\title[Entropy Geometry and Normalized Means on Infinite-Dimensional Hamiltonian Manifolds]
{Entropy Geometry and Normalized Means on Infinite-Dimensional Hamiltonian Manifolds}

\author[Jean-Pierre Magnot]{Jean-Pierre Magnot}

\address{SFR MATHSTIC, LAREMA, Universit\'e d'Angers,
2 Bd Lavoisier, 49045 Angers cedex 1, France;
Lyc\'ee Jeanne d'Arc, 40 avenue de Grande Bretagne,
63000 Clermont-Ferrand, France;
Lepage Research Institute, 17 novembra 1, 081 16 Presov, Slovakia}

\email{\small magnot@math.cnrs.fr;
jean-pierr.magnot@ac-clermont.fr}

\date{\today}

\begin{document}

\maketitle

\begin{abstract}
We propose a geometric--analytic framework for equilibrium statistical
mechanics on infinite-dimensional Hamiltonian systems. In situations
where no suitable $\sigma$-additive invariant measure is available, we
use \emph{normalized means}, which generalize probability measures and
normalized integrals. This construction yields entropy and free-energy
functionals on weak symplectic Fréchet manifolds and gives existence and
uniqueness of exponential-family equilibrium states under explicit
admissibility and separation assumptions. These states are stationary
under Hamiltonian flows preserving both the reference mean and the
equilibrium weight, and satisfy a classical Poisson--KMS identity when
the reference mean is Poisson invariant. Under a local exponential
regularity assumption, the logarithmic partition functional is smooth
and convex, with Hessian given by the covariance form. It is strictly
convex modulo thermodynamically null directions and, through
Legendre--Fenchel duality, induces a concave entropy on the domain of
extensive variables. We illustrate the framework with $H^s$-geodesic
equations on current groups $\operatorname{Map}(M,G)$ and
diffeomorphism groups $\operatorname{Diff}(M)$, including hydrodynamic
and field-theoretic examples.
\end{abstract}



\section*{Introduction}

The idea that entropy should be regarded not merely as a numerical
measure of disorder, but as a genuinely geometric object, lies at the
heart of J.-M.~Souriau's formulation of statistical mechanics
\cite{Souriau70}. In this approach, a Hamiltonian system
\[
(\mathcal M,\omega,H)
\]
is considered together with its symmetries and moment map, and
thermodynamic equilibrium is described through a variational principle
on a space of statistical states. In finite dimension, this construction
leads naturally to exponential families, while the logarithm of the
partition function provides the potential from which the principal
thermodynamic quantities are derived.

Souriau's point of view differs from the traditional interpretation of
entropy as either a combinatorial count or a functional defined on a
prescribed space of probability measures. Entropy becomes instead a
potential associated with the symplectic and group-theoretic structure
of the system. Coadjoint orbits, equipped with their
Kostant--Kirillov--Souriau forms, provide the natural geometric setting,
and the passage between intensive and extensive variables is governed
by convex duality. Modern developments of this theory
\cite{Barbaresco09,Barbaresco25,IglesiasSouriauBook,Marle16,Neeb26}
have revealed its close connections with information geometry,
Lie-group thermodynamics, moment-map geometry, and the structure of
coadjoint orbits.

The purpose of the present paper, whose general programme was first
sketched in \cite{MagnotAOP26}, is to extend this geometric formulation
to infinite-dimensional Hamiltonian systems. The principal obstruction
is familiar: on an infinite-dimensional phase space there is, in
general, no analogue of Lebesgue measure, and a suitable
$\sigma$-additive invariant measure may fail to exist. This difficulty
appears naturally in Hamiltonian PDEs, hydrodynamic models, and field
theories. Depending on the system, one may introduce a Gaussian or
renormalised reference measure, but in many situations the resulting
construction remains formal, depends on a regularisation procedure, or
does not yield an invariant probability measure on the original smooth
phase space.

Our approach is to replace probability measures by
\emph{normalized means}. Such objects were studied in
\cite{Magnot17} as limits of normalized integrals over
finite-dimensional approximations. At the algebraic level, a normalized
mean on a vector lattice algebra $\mathcal L$ is a positive normalized
linear functional
\[
\mathsf m\colon\mathcal L\longrightarrow\mathbb R,
\qquad
\mathsf m(1)=1,
\qquad
\mathsf m(f)\geq 0
\quad\text{whenever }f\geq 0.
\]
Unlike an ordinary probability measure, it need not arise from a
$\sigma$-additive set function. In this respect, normalized means extend
to general configuration spaces the invariant means that occur in
amenable group theory
\cite{vonNeumann29,Day57,Paterson88}. They also provide a natural
language for describing limits of normalized finite-dimensional
cut-offs without postulating the existence of a fictitious
infinite-dimensional Lebesgue measure.

A technical distinction is necessary from the outset. Bounded test
observables belong to a unital algebra
\[
\mathcal A_b\subset C_b(\mathcal M),
\]
whereas the Hamiltonian and the constraint observables are generally
unbounded and belong to a larger vector lattice algebra $\mathcal L$.
Only those thermodynamic potentials whose exponential weights lie in
the domain of the reference mean are retained. This leads to an
admissible parameter domain on which the partition functional is both
finite and strictly positive.

Given a reference mean $\mathsf m_0$, we compare another normalized mean
$\mathsf n$ with $\mathsf m_0$ through the variational relative entropy
\[
\mathcal H(\mathsf n\,\|\,\mathsf m_0)
=
\sup_{f\in\mathcal E(\mathsf m_0)}
\left\{
\mathsf n(f)
-
\log\mathsf m_0(e^f)
\right\},
\]
where $\mathcal E(\mathsf m_0)$ denotes the class of exponentially
admissible potentials. This is the natural analogue, in the present
setting, of the Donsker--Varadhan variational representation. For a
Hamiltonian $H$ and a finite family of constraints
$\Phi_1,\ldots,\Phi_N$, the corresponding thermodynamic potential is
\[
F_{\beta,\lambda}
=
-\beta H-\sum_{a=1}^{N}\lambda^a\Phi_a,
\]
and the logarithmic partition functional is
\[
\psi(\beta,\lambda)
=
\log\mathsf m_0(e^{F_{\beta,\lambda}}).
\]

The Gibbs variational principle then takes the form
\[
-\log Z(\beta,\lambda)
=
\inf_{\mathsf n}
\left\{
\mathcal H(\mathsf n\,\|\,\mathsf m_0)
+
\beta\,\mathsf n(H)
+
\sum_{a=1}^{N}\lambda^a\mathsf n(\Phi_a)
\right\}.
\]
Under explicit admissibility, compatibility, and separation
assumptions, the minimizer exists, is unique, and is given by the
exponential tilt
\[
\mathsf n_{\beta,\lambda}^*(f)
=
\frac{
\mathsf m_0\left(
f e^{-\beta H-\sum_{a=1}^{N}\lambda^a\Phi_a}
\right)
}{
\mathsf m_0\left(
e^{-\beta H-\sum_{a=1}^{N}\lambda^a\Phi_a}
\right)
}.
\]
Thus the familiar Gibbs prescription survives without requiring the
reference state to be represented by a $\sigma$-additive probability
measure.

The geometry of the resulting exponential family is encoded by the
partition potential. Under a local exponential regularity hypothesis,
$\psi$ is smooth, and its Hessian is the covariance form of the
Hamiltonian and the constraint observables. It is therefore positive
semidefinite and induces a positive-definite quadratic form on the
quotient of each tangent space by the kernel of the covariance form.
When these kernels have locally constant rank, the resulting quotients
define a Riemannian metric on the corresponding reduced parameter
space. In particular, strict convexity holds whenever no nonzero
parameter direction is thermodynamically null throughout a nontrivial
parameter segment. The corresponding entropy is obtained through the
concave Legendre--Fenchel dual
\[
\mathcal S(E,c)
=
\inf_{(\beta,\lambda)}
\left\{
\psi(\beta,\lambda)
+
\beta E
+
\sum_{a=1}^{N}\lambda^a c_a
\right\}.
\]
On the image of the equilibrium map, this dual coincides with the
constrained variational entropy and gives the expected correspondence
between intensive and extensive variables.

The interaction between equilibrium and Hamiltonian dynamics requires
some care. If the Hamiltonian flow preserves the reference mean and
each constraint observable $\Phi_a$, then the exponentially tilted
equilibrium mean is stationary. Stationarity alone does not imply the
complex-time KMS boundary condition familiar from noncommutative
operator algebras. There is, however, a natural classical counterpart.
When the reference mean is invariant under the Poisson bracket in the
sense of an integration-by-parts identity, the equilibrium mean
satisfies the Poisson--KMS relation
\[
\mathsf n_{\beta,\lambda}^*(\{f,g\})
=
\beta\,
\mathsf n_{\beta,\lambda}^*
\bigl(g\{f,H\}\bigr)
+
\sum_{a=1}^{N}\lambda^a
\mathsf n_{\beta,\lambda}^*
\bigl(g\{f,\Phi_a\}\bigr).
\]
This identity expresses thermal equilibrium directly in terms of the
classical Hamiltonian structure, without appealing to a formal
complexification of the flow.

The abstract construction is discussed in connection with
right-invariant $H^s$-geodesic equations on current groups
$\operatorname{Map}(M,G)$ and on diffeomorphism groups
$\operatorname{Diff}(M)$. These examples include Euler--Arnold systems,
EPDiff, and the Camassa--Holm equation. Two-dimensional Euler dynamics
and its Casimir constraints lead, at the formal mean-field level, to
relations of Miller--Robert--Sommeria type. These applications connect
the present framework with the geometry of ideal fluids developed in
\cite{EbinMarsden70}, with vortex statistical mechanics
\cite{MillerRobertSommeria91}, and with the construction of invariant
measures for Hamiltonian PDEs initiated in particular by
\cite{Bourgain94}. Their purpose is both to show how the abstract
hypotheses arise in concrete systems and to identify the additional
analytic estimates required for a complete model-specific
construction.

The paper is organised as follows. Section~1 introduces weak
Hamiltonian Fréchet geometry, normalized means, exponentially admissible
potentials, relative entropy, the covariance Hessian, and the relevant
Legendre--Fenchel duality. Section~2 establishes the existence and
uniqueness of equilibrium means, studies their stationarity, and proves
the classical Poisson--KMS identity under the appropriate invariance
hypothesis. Sections~3 and~4 develop the current-group,
diffeomorphism-group, hydrodynamic, and Fourier-mode examples.
Section~5 collects the full proofs of the principal statements and
clarifies the compatibility conditions required by exponential
tilting.

\section{Infinite-Dimensional Hamiltonian Geometry and Normalized Means}

This section introduces the functional and geometric framework used
throughout the paper.

\subsection{Hamiltonian geometry on Fréchet manifolds}

Let $\mathcal M$ be a Fréchet manifold modelled on a locally convex
topological vector space $E$. We assume that $\mathcal M$ is equipped
with a smooth closed $2$-form $\omega$ such that, for every
$m\in\mathcal M$, the induced map
\[
\omega_m^\flat\colon T_m\mathcal M\longrightarrow T_m^*\mathcal M,
\qquad
v\longmapsto\omega_m(v,\cdot),
\]
is continuous and injective. Thus $\omega$ is a \emph{weak symplectic
form}. Writing
\[
J_m=\omega_m^\flat,
\]
we have
\[
\omega_m(v,w)=\langle J_m v,w\rangle.
\]
The pair $(\mathcal M,\omega)$ is called a \emph{weak symplectic
Fréchet manifold} (see \cite{Hamilton82,Neeb06,KM97}).

A \emph{Hamiltonian function} is a smooth map
\[
H\colon\mathcal M\longrightarrow\mathbb R
\]
such that
\begin{equation}
\omega\bigl(X_H,\cdot\bigr)=\mathrm dH
\label{eq:Hamiltonian}
\end{equation}
admits a smooth solution $X_H$. The associated vector field generates a
possibly local Hamiltonian flow $\Phi_t^H$.

Let $G$ be a Fréchet--Lie group acting smoothly on $\mathcal M$ by
symplectomorphisms, with Lie algebra $\mathfrak g$. A \emph{moment map}
is a smooth map
\[
\mathbf J\colon\mathcal M\longrightarrow\mathfrak g^*
\]
such that, for every $\xi\in\mathfrak g$, the function
\[
m\longmapsto\langle\mathbf J(m),\xi\rangle
\]
is a Hamiltonian generating the infinitesimal action of $\xi$. This
extends the classical finite-dimensional theory
\cite{AbrahamMarsden78,MarsdenRatiu99,Omori97}.

\begin{Definition}[Cylindrical functional]
A function $f\colon\mathcal M\to\mathbb R$ is called cylindrical if it
factors through a finite-dimensional smooth map, i.e.
\[
f=\widetilde f\circ\pi_N
\]
for some smooth map
\[
\pi_N\colon\mathcal M\longrightarrow\mathbb R^N
\]
and some function $\widetilde f$ on $\mathbb R^N$.
\end{Definition}

\subsection{Normalized means}

We now introduce the notion replacing probability measures in infinite
dimension.

Let
\[
\mathcal A_b\subset C_b(\mathcal M)
\]
be a unital algebra of bounded continuous observables, stable under the
exponential operations considered below. Let $\mathcal L$ be a unital
vector lattice algebra of real-valued functions on $\mathcal M$ such
that
\[
\mathcal A_b\subset\mathcal L.
\]
The space $\mathcal L$ may contain unbounded functions, including the
Hamiltonians and constraint observables used in the applications.

\begin{Definition}[Normalized mean]
\label{def:normalized-mean}
A \emph{normalized mean} on $\mathcal L$ is a linear functional
\[
\mathsf m\colon\mathcal L\longrightarrow\mathbb R
\]
such that:
\begin{enumerate}
\item[\textup{(i)}] $\mathsf m(1)=1$;
\item[\textup{(ii)}] $\mathsf m(f)\geq 0$ for every $f\in\mathcal L$
such that $f\geq 0$.
\end{enumerate}
\end{Definition}

We denote by $\mathsf S(\mathcal L)$ the set of all normalized means on
$\mathcal L$.

Normalized means encode finite additivity on every family of sets whose
indicator functions belong to $\mathcal L$, but they do not rely on
$\sigma$-additivity.

\begin{Definition}[Invariance]
A normalized mean $\mathsf m_0$ is said to be invariant under a flow
$\Phi_t$ if, for all $f\in\mathcal L$,
\[
\mathsf m_0(f\circ\Phi_t)=\mathsf m_0(f)
\]
whenever $f\circ\Phi_t\in\mathcal L$.
\end{Definition}

They generalize invariant means from amenable group theory.

\begin{rem}[Relation with amenability]
Invariant means on groups were introduced by von Neumann and further
developed by Day and Paterson
\cite{vonNeumann29,Day57,Paterson88}. The present framework extends this
concept to general configuration spaces.
\end{rem}

\begin{example}[Cut-off construction]
\label{ex:cutoff-mean}
Let $(B_N,\mu_N)$ be a sequence of finite-dimensional approximations of
$\mathcal M$, with
\[
0<\int_{B_N}\mathrm d\mu_N<\infty.
\]
If, for every $f\in\mathcal L$ under consideration, the limit
\[
\mathsf m_0(f)
=
\lim_{N\to\infty}
\frac{\displaystyle\int_{B_N}f\,\mathrm d\mu_N}
{\displaystyle\int_{B_N}\mathrm d\mu_N}
\]
exists, is finite, and is independent of the admissible truncation, then
$\mathsf m_0$ defines a normalized mean on its domain
\cite{Magnot17}.
\end{example}

\subsection{Topology on the space of normalized means}
\label{subsec:topology-means}

The definition of a normalized mean is algebraic: positivity and
normalization do not require continuity with respect to the compact-open
topology. This distinction is important for means obtained from
finite-dimensional cut-offs, which need not be continuous for that
topology.

We equip $\mathsf S(\mathcal L)$ with the topology of pointwise
convergence on $\mathcal L$. Thus, a net
$(\mathsf m_i)_{i\in I}$ converges to $\mathsf m$ if and only if
\[
\mathsf m_i(f)\longrightarrow\mathsf m(f)
\qquad
\text{for every }f\in\mathcal L.
\]
Equivalently, this is the weakest topology for which all evaluation maps
\[
\operatorname{ev}_f\colon
\mathsf S(\mathcal L)\longrightarrow\mathbb R,
\qquad
\operatorname{ev}_f(\mathsf m)=\mathsf m(f),
\]
are continuous.

When $\mathcal A_b\subset C_b(\mathcal M)$ is the algebra of bounded
test observables, positivity and normalization imply
\[
\left|\mathsf m(f)\right|
\leq
\|f\|_\infty,
\qquad
f\in\mathcal A_b,
\quad
\mathsf m\in\mathsf S(\mathcal A_b).
\]
Indeed,
\[
-\|f\|_\infty 1
\leq
f
\leq
\|f\|_\infty 1,
\]
and positivity gives the required estimate.

Consequently, $\mathsf S(\mathcal A_b)$ may be identified with a closed
subset of
\[
\prod_{f\in\mathcal A_b}
[-\|f\|_\infty,\|f\|_\infty].
\]
By Tychonoff's theorem, it is compact for the topology of pointwise
convergence on $\mathcal A_b$.

\begin{Proposition}
\label{prop:compact-bounded-states}
The space $\mathsf S(\mathcal A_b)$ of normalized means on the bounded
observable algebra $\mathcal A_b$ is compact for the topology of
pointwise convergence on $\mathcal A_b$.
\end{Proposition}

\begin{proof}
For every $f\in\mathcal A_b$, positivity and normalization give
\[
|\mathsf m(f)|\leq\|f\|_\infty.
\]
Hence $\mathsf S(\mathcal A_b)$ is contained in the compact product
\[
K
=
\prod_{f\in\mathcal A_b}
[-\|f\|_\infty,\|f\|_\infty].
\]
Linearity, positivity, and normalization are preserved under pointwise
limits. Therefore $\mathsf S(\mathcal A_b)$ is closed in $K$, and is
thus compact.
\end{proof}

For the larger algebra $\mathcal L$, which may contain unbounded
observables, no analogous compactness statement holds in general.
Accordingly, compactness of the entire space $\mathsf S(\mathcal L)$
will not be assumed. Variational arguments based on compactness require
instead compactness of the relevant free-energy sublevel sets,
\[
\left\{
\mathsf n\in\mathsf S(\mathcal L)
\;\middle|\;
\mathcal G_{\beta,\lambda}(\mathsf n)\leq r
\right\},
\]
for the topology of pointwise convergence. Whenever such an argument is
used, this property will be included explicitly among the coercivity
assumptions.

\subsection{Pointwise limits and cut-off normalized means}
\label{subsec:cutoff-limits}

The topology of pointwise convergence is naturally adapted to
normalized means constructed from finite-dimensional approximations.

Let $(B_N,\mu_N)$ be a sequence of finite-dimensional approximations
such that
\[
0<\int_{B_N}\mathrm d\mu_N<\infty,
\]
and define
\[
\mathsf m_N(f)
=
\frac{
\displaystyle\int_{B_N}f\,\mathrm d\mu_N
}{
\displaystyle\int_{B_N}\mathrm d\mu_N
}
\]
whenever the quotient is well defined.

\begin{Proposition}
\label{prop:cutoff-pointwise-limit}
Assume that, for every $f\in\mathcal A_b$, the limit
\[
\mathsf m_0(f)
=
\lim_{N\to\infty}\mathsf m_N(f)
\]
exists. Then $\mathsf m_0$ is a normalized mean on $\mathcal A_b$.

More generally, if the same pointwise limit exists and is finite for
every $f$ in a vector lattice algebra $\mathcal L$, then $\mathsf m_0$
is a normalized mean on $\mathcal L$.
\end{Proposition}

\begin{proof}
For every $N$, the functional $\mathsf m_N$ is linear, positive, and
normalized. These three properties are preserved under pointwise
limits. In particular,
\[
\mathsf m_0(1)
=
\lim_{N\to\infty}\mathsf m_N(1)
=
1,
\]
and, whenever $f\geq 0$,
\[
\mathsf m_0(f)
=
\lim_{N\to\infty}\mathsf m_N(f)
\geq 0.
\]
Therefore $\mathsf m_0$ is a normalized mean.
\end{proof}

\begin{rem}
No continuity for the compact-open topology is required in
Proposition~\ref{prop:cutoff-pointwise-limit}. In fact, means describing
asymptotic behavior at infinity are generally not continuous for that
topology.
\end{rem}

\begin{rem}
The preceding proposition proves the existence of a normalized mean
once pointwise convergence has been established. It does not by itself
prove that the resulting mean is independent of the chosen
finite-dimensional exhaustion. Such independence must either be
verified in each application or included in the definition of an
admissible family of cut-offs.
\end{rem}

\subsection{Exponentially admissible potentials}

The Hamiltonians arising in the applications are generally unbounded.
They are therefore treated as thermodynamic potentials rather than as
elements of the algebra $\mathcal A_b$ of bounded test observables.

\begin{Definition}[Exponentially admissible potential]
\label{def:admissible-potential}
A real-valued function $F\in\mathcal L$ is called
\emph{$\mathsf m_0$-admissible} if
\[
e^F\in\mathcal L,
\qquad
0<\mathsf m_0(e^F)<\infty,
\]
and
\[
fe^F\in\mathcal L
\qquad
\text{for every }f\in\mathcal L.
\]
The set of $\mathsf m_0$-admissible potentials is denoted by
$\mathcal E(\mathsf m_0)$.
\end{Definition}

We assume that $\mathcal E(\mathsf m_0)$ contains $0$ and is stable
under the bounded perturbations used below. For
$F\in\mathcal E(\mathsf m_0)$, define
\[
\Lambda_{\mathsf m_0}(F)
=
\log\mathsf m_0(e^F).
\]

\begin{Definition}[Exponential tilt]
\label{def:exponential-tilt}
Let $F\in\mathcal E(\mathsf m_0)$. The exponential tilt of
$\mathsf m_0$ by $F$ is the normalized mean
\[
\mathsf m_F(f)
=
\frac{\mathsf m_0(fe^F)}{\mathsf m_0(e^F)},
\qquad
f\in\mathcal L.
\]
\end{Definition}

The positivity and normalization of $\mathsf m_F$ follow immediately:
\[
\mathsf m_F(f)\geq 0
\quad\text{whenever }f\geq 0,
\qquad
\mathsf m_F(1)=1.
\]

\subsection{Relative entropy}

Given $\mathsf m_0\in\mathsf S(\mathcal L)$, we define entropy through a
variational principle.

\begin{Definition}[Relative entropy]
\label{def:entropy}
For $\mathsf n,\mathsf m_0\in\mathsf S(\mathcal L)$, define
\[
\mathcal H(\mathsf n\,\|\,\mathsf m_0)
=
\sup_{f\in\mathcal E(\mathsf m_0)}
\left\{
\mathsf n(f)-\log\mathsf m_0(e^f)
\right\},
\]
with values in $[0,+\infty]$.
\end{Definition}

This definition is inspired by the Donsker--Varadhan variational
formula.

\begin{Proposition}[Basic properties of the relative entropy]
\label{prop:entropy-properties}
Let $\mathsf m_0$ be a normalized mean on $\mathcal L$. Assume that
there exists a vector subspace $\mathcal T\subset\mathcal L$ separating
normalized means such that, for every $h\in\mathcal T$, one has
$th\in\mathcal E(\mathsf m_0)$ for all sufficiently small
$t\in\mathbb R$, and
\[
\left.\frac{\mathrm d}{\mathrm dt}\right|_{t=0}
\log\mathsf m_0(e^{th})
=
\mathsf m_0(h).
\]
Then:
\begin{enumerate}
\item
\[
\mathcal H(\mathsf n\,\|\,\mathsf m_0)\geq 0,
\]
with equality if and only if $\mathsf n=\mathsf m_0$;
\item the map
\[
\mathsf n\longmapsto
\mathcal H(\mathsf n\,\|\,\mathsf m_0)
\]
is convex and lower semicontinuous for the topology of pointwise
convergence on $\mathsf S(\mathcal L)$.
\end{enumerate}
\end{Proposition}

\begin{proof}
The non-negativity follows by taking $f=0$. Convexity and lower
semicontinuity follow because
\[
\mathsf n\longmapsto
\mathsf n(f)-\log\mathsf m_0(e^f)
\]
is affine and continuous for the topology of pointwise convergence, for
every admissible $f$.

Suppose now that
\[
\mathcal H(\mathsf n\,\|\,\mathsf m_0)=0.
\]
For every $h\in\mathcal T$ and every sufficiently small
$t\in\mathbb R$, one has
\[
t\mathsf n(h)
\leq
\log\mathsf m_0(e^{th}).
\]
Dividing by $t$ and letting $t\to0$ from the positive and negative sides
gives
\[
\mathsf n(h)=\mathsf m_0(h).
\]
Since $\mathcal T$ separates normalized means, it follows that
$\mathsf n=\mathsf m_0$.
\end{proof}

\subsection{Change-of-reference compatibility}

The variational entropy is defined through a reference-dependent class
of admissible potentials. Consequently, exponential tilting requires a
compatibility condition between the admissible classes associated with
the original and tilted means.

\begin{Definition}[Change-of-reference compatibility]
\label{def:change-reference-compatible}
Let $F\in\mathcal E(\mathsf m_0)$. We say that the exponential tilt by
$F$ is \emph{change-of-reference compatible} if
\[
\mathcal H(\mathsf n\,\|\,\mathsf m_F)
=
\mathcal H(\mathsf n\,\|\,\mathsf m_0)
-\mathsf n(F)
+\log\mathsf m_0(e^F)
\]
for every normalized mean $\mathsf n$ for which the terms are defined.
\end{Definition}

A sufficient translation condition on the admissible classes, together
with a proof of the resulting change-of-reference identity, is given in
section~\ref{subsec:change-reference}.

\subsection{Analytic assumptions}

Let $H,\Phi_1,\ldots,\Phi_N\in\mathcal L$ be possibly unbounded
observables, and set
\[
F_{\beta,\lambda}
=
-\beta H-\sum_{a=1}^{N}\lambda^a\Phi_a.
\]
The admissible thermodynamic domain is
\[
\mathcal D
=
\left\{
(\beta,\lambda)\in\mathbb R^{N+1}
\;\middle|\;
F_{\beta,\lambda}\in\mathcal E(\mathsf m_0)
\right\}.
\]

Depending on the result under consideration, we shall use the following
assumptions:

\begin{itemize}
\item[(A1)] Exponential admissibility: the relevant parameters belong
to $\mathcal D$, so that
\[
0<
\mathsf m_0\left(
e^{-\beta H-\sum_{a=1}^{N}\lambda^a\Phi_a}
\right)
<\infty;
\]

\item[(A2)] Coercivity: the sublevel sets of
$\mathcal G_{\beta,\lambda}$ are compact in
$\mathsf S(\mathcal L)$ for the topology of pointwise convergence;

\item[(A3)] Non-degeneracy: the covariance form has no nonzero null
direction on the reduced parameter domain under consideration.
\end{itemize}

\subsection{Variational principle}

For $(\beta,\lambda)\in\mathcal D$, define
\[
Z(\beta,\lambda)
=
\mathsf m_0\left(e^{F_{\beta,\lambda}}\right).
\]

The free-energy functional associated with $(\beta,\lambda)$ is
\[
\mathcal G_{\beta,\lambda}(\mathsf n)
=
\mathcal H(\mathsf n\,\|\,\mathsf m_0)
+
\beta\mathsf n(H)
+
\sum_{a=1}^{N}\lambda^a\mathsf n(\Phi_a),
\qquad
\mathsf n\in\mathsf S(\mathcal L).
\]
Equivalently,
\[
\mathcal G_{\beta,\lambda}(\mathsf n)
=
\mathcal H(\mathsf n\,\|\,\mathsf m_0)
-
\mathsf n(F_{\beta,\lambda}).
\]

\begin{Proposition}[Gibbs variational principle]
\label{prop:gibbs-variational-principle}
Let $(\beta,\lambda)\in\mathcal D$ and assume that the exponential tilt
by $F_{\beta,\lambda}$ is change-of-reference compatible. Assume also
that the relative entropy with reference
$\mathsf n_{\beta,\lambda}^*$ has the separation property of
Proposition~\ref{prop:entropy-properties}. Then
\[
-\log Z(\beta,\lambda)
=
\inf_{\mathsf n\in\mathsf S(\mathcal L)}
\mathcal G_{\beta,\lambda}(\mathsf n).
\]
Equivalently,
\[
\log Z(\beta,\lambda)
=
\sup_{\mathsf n\in\mathsf S(\mathcal L)}
\left\{
\mathsf n(F_{\beta,\lambda})
-
\mathcal H(\mathsf n\,\|\,\mathsf m_0)
\right\}.
\]
The extremum is attained at the unique normalized mean
\[
\mathsf n_{\beta,\lambda}^{*}(f)
=
\frac{
\mathsf m_0\left(fe^{F_{\beta,\lambda}}\right)
}{
\mathsf m_0\left(e^{F_{\beta,\lambda}}\right)
},
\qquad
f\in\mathcal L.
\]
\end{Proposition}

\begin{proof}
By the change-of-reference identity,
\[
\mathcal H\left(
\mathsf n\,\|\,\mathsf n_{\beta,\lambda}^{*}
\right)
=
\mathcal H(\mathsf n\,\|\,\mathsf m_0)
-\mathsf n(F_{\beta,\lambda})
+\log Z(\beta,\lambda).
\]
Therefore
\[
\mathcal H\left(
\mathsf n\,\|\,\mathsf n_{\beta,\lambda}^{*}
\right)
=
\mathcal G_{\beta,\lambda}(\mathsf n)
+\log Z(\beta,\lambda).
\]
Since relative entropy is nonnegative,
\[
\mathcal G_{\beta,\lambda}(\mathsf n)
\geq
-\log Z(\beta,\lambda).
\]
For
\[
\mathsf n=\mathsf n_{\beta,\lambda}^{*},
\]
the relative entropy on the left-hand side vanishes, and hence equality
holds. If equality holds for another normalized mean $\mathsf n$, then
\[
\mathcal H\left(
\mathsf n\,\|\,\mathsf n_{\beta,\lambda}^{*}
\right)=0.
\]
The separation property of the relative entropy with reference
$\mathsf n_{\beta,\lambda}^{*}$ therefore implies
\[
\mathsf n=\mathsf n_{\beta,\lambda}^{*}.
\]
\end{proof}

\begin{rem}
The compactness assumption \textup{(A2)} is not needed in the preceding
proposition because the minimizer is explicitly constructed by
exponential tilting. It becomes relevant for existence arguments in
which no compatible exponential tilt is assumed in advance.
\end{rem}

\subsection{Entropy geometry}

We now relate the previous construction to symplectic geometry.

For a left action of $G$ on $\mathcal M$, we use the induced action on
observables
\[
(g\cdot f)(m)=f(g^{-1}\cdot m).
\]

\begin{Theorem}[Geometric equilibrium for normalized means]
\label{thm:geom-eq}
Let $(\mathcal M,\omega,G,\mathbf J)$ be a weak Hamiltonian Fréchet
$G$-space, and let $\mathsf m_0$ be a $G$-invariant normalized mean on
$\mathcal L$. Assume that the action of $G$ preserves $\mathcal L$.

Assume that $H\in\mathcal L$ is $G$-invariant and that the moment map
\[
\mathbf J\colon\mathcal M\longrightarrow\mathfrak g^*
\]
is equivariant. For $\lambda\in\mathfrak g$ and
$(\beta,\lambda)$ in the admissible domain, define
\[
F_{\beta,\lambda}
=
-\beta H-\langle\mathbf J,\lambda\rangle
\]
and
\[
Z(\beta,\lambda)
=
\mathsf m_0\left(e^{F_{\beta,\lambda}}\right).
\]
The corresponding equilibrium mean is
\[
\mathsf n_{\beta,\lambda}^{*}(f)
=
\frac{
\mathsf m_0\left(fe^{F_{\beta,\lambda}}\right)
}{
\mathsf m_0\left(e^{F_{\beta,\lambda}}\right)
},
\qquad
f\in\mathcal L.
\]

Then the family of equilibrium means is equivariant with respect to the
adjoint action on the parameter:
\[
\mathsf n_{\beta,\lambda}^{*}(g\cdot f)
=
\mathsf n_{\beta,\operatorname{Ad}_{g^{-1}}\lambda}^{*}(f),
\qquad
g\in G.
\]
In particular, $\mathsf n_{\beta,\lambda}^{*}$ is invariant under the
stabilizer
\[
G_\lambda
=
\left\{
g\in G
\;\middle|\;
\operatorname{Ad}_{g^{-1}}\lambda=\lambda
\right\}.
\]
If $\lambda$ is fixed by the adjoint action, then
$\mathsf n_{\beta,\lambda}^{*}$ is $G$-invariant.

Moreover, the thermodynamic potential
\[
\psi(\beta,\lambda)=\log Z(\beta,\lambda)
\]
generates the correspondence between the intensive variables
$(\beta,\lambda)$ and the extensive variables
\[
E(\beta,\lambda)
=
\mathsf n_{\beta,\lambda}^{*}(H),
\qquad
c(\beta,\lambda)
=
\mathsf n_{\beta,\lambda}^{*}(\mathbf J),
\]
on every domain on which $\psi$ is differentiable.
\end{Theorem}

\begin{proof}
Equivariance of the moment map means
\[
\mathbf J(g\cdot m)=\operatorname{Ad}_g^*\mathbf J(m).
\]
Consequently,
\[
F_{\beta,\lambda}(g\cdot m)
=
-\beta H(m)
-
\left\langle
\operatorname{Ad}_g^*\mathbf J(m),\lambda
\right\rangle
=
F_{\beta,\operatorname{Ad}_{g^{-1}}\lambda}(m).
\]
Using the $G$-invariance of $\mathsf m_0$, we obtain
\[
\begin{split}
\mathsf n_{\beta,\lambda}^{*}(g\cdot f)
&=
\frac{
\mathsf m_0\left((g\cdot f)e^{F_{\beta,\lambda}}\right)
}{
\mathsf m_0\left(e^{F_{\beta,\lambda}}\right)
}
\\
&=
\frac{
\mathsf m_0\left(
f e^{F_{\beta,\operatorname{Ad}_{g^{-1}}\lambda}}
\right)
}{
\mathsf m_0\left(
e^{F_{\beta,\operatorname{Ad}_{g^{-1}}\lambda}}
\right)
}
\\
&=
\mathsf n_{\beta,\operatorname{Ad}_{g^{-1}}\lambda}^{*}(f).
\end{split}
\]
The assertions concerning the stabilizer follow immediately. The
correspondence between intensive and extensive variables follows from
the first-derivative formulas established below.
\end{proof}

\begin{rem}
The preceding theorem establishes equivariance, and in particular
invariance under the stabilizer of the thermodynamic parameter. A
KMS-type property requires additional Poisson-algebraic assumptions and
does not follow from invariance alone.
\end{rem}

\subsection{Local exponential regularity}
\label{subsec:local-exponential-regularity}

Pointwise finiteness of the partition functional does not by itself
imply differentiability with respect to the thermodynamic parameters.
We therefore introduce a local exponential moment condition.

Let
\[
\theta=(\beta,\lambda^1,\ldots,\lambda^N)\in\mathcal D
\]
and set
\[
\mathcal O_0=H,
\qquad
\mathcal O_a=\Phi_a,
\quad
a=1,\ldots,N.
\]

\begin{Definition}[Local exponential regularity]
\label{def:local-exponential-regularity}
The partition functional is said to be
\emph{locally exponentially regular at $\theta\in\mathcal D$} if
$\theta$ is an interior point of $\mathcal D$ and there exists
$\varepsilon>0$ such that
\[
\exp\left(
F_\theta+\varepsilon\sum_{j=0}^{N}|\mathcal O_j|
\right)
\in\mathcal L
\]
and
\[
\mathsf m_0\left(
\exp\left(
F_\theta+\varepsilon\sum_{j=0}^{N}|\mathcal O_j|
\right)
\right)
<\infty.
\]
We say that the partition functional is locally exponentially regular
on $\mathcal D$ if this condition holds at every
$\theta\in\mathcal D$.
\end{Definition}

The local exponential regularity condition provides a common
exponential bound for sufficiently small parameter increments. Indeed,
if
\[
\delta\theta
=
(\delta\beta,\delta\lambda^1,\ldots,\delta\lambda^N)
\]
satisfies
\[
|\delta\beta|<\varepsilon,
\qquad
|\delta\lambda^a|<\varepsilon,
\quad
a=1,\ldots,N,
\]
then
\[
F_{\theta+\delta\theta}
\leq
F_\theta
+
\varepsilon\sum_{j=0}^{N}|\mathcal O_j|.
\]

\begin{Proposition}[Smoothness of the partition functional]
\label{prop:partition-smoothness}
Assume that the partition functional is locally exponentially regular
on $\mathcal D$. Then $\mathcal D$ is open and
\[
Z(\theta)=\mathsf m_0(e^{F_\theta})
\]
is smooth on $\mathcal D$. Since $Z(\theta)>0$, the logarithmic
partition functional
\[
\psi(\theta)=\log Z(\theta)
\]
is also smooth.

For every parameter direction $\dot\theta$, one has
\[
DZ(\theta)[\dot\theta]
=
-\mathsf m_0\left(
\mathcal O_{\dot\theta}e^{F_\theta}
\right),
\]
where
\[
\mathcal O_{\dot\theta}
=
\dot\beta H+\sum_{a=1}^{N}\dot\lambda^a\Phi_a.
\]
Consequently,
\[
D\psi(\theta)[\dot\theta]
=
-\mathsf n_\theta^*(\mathcal O_{\dot\theta}).
\]

More generally, for parameter directions
$\dot\theta_1,\ldots,\dot\theta_k$,
\[
D^kZ(\theta)
[\dot\theta_1,\ldots,\dot\theta_k]
=
(-1)^k
\mathsf m_0\left(
\mathcal O_{\dot\theta_1}\cdots
\mathcal O_{\dot\theta_k}
e^{F_\theta}
\right).
\]
The higher derivatives of $\psi$ are the corresponding joint cumulants
with respect to $\mathsf n_\theta^*$.
\end{Proposition}

\begin{proof}
The openness of $\mathcal D$ follows from
Definition~\ref{def:local-exponential-regularity}. Fix
$\theta\in\mathcal D$. The local exponential regularity condition
provides a common exponential bound for the parameter increments.

For a direction $\dot\theta$ and sufficiently small $t$, Taylor's
formula gives
\[
e^{F_{\theta+t\dot\theta}}
=
e^{F_\theta}
\left(
1-t\mathcal O_{\dot\theta}
+
R_2(t)
\right),
\]
where
\[
|R_2(t)|
\leq
\frac{t^2}{2}
\mathcal O_{\dot\theta}^{\,2}
e^{|t|\,|\mathcal O_{\dot\theta}|}.
\]
After decreasing the parameter neighborhood if necessary, the
right-hand side is bounded by a fixed element of $\mathcal L$ having a
finite $\mathsf m_0$-mean. Positivity of $\mathsf m_0$ and the explicit
factor $t^2$ therefore imply
\[
DZ(\theta)[\dot\theta]
=
-\mathsf m_0\left(
\mathcal O_{\dot\theta}e^{F_\theta}
\right).
\]

The same argument, applied to higher-order Taylor remainders, gives
\[
D^kZ(\theta)
[\dot\theta_1,\ldots,\dot\theta_k]
=
(-1)^k
\mathsf m_0\left(
\mathcal O_{\dot\theta_1}\cdots
\mathcal O_{\dot\theta_k}
e^{F_\theta}
\right).
\]
The common local exponential bound also gives continuity of these
derivatives. Hence $Z$ is smooth. Since $Z>0$ on $\mathcal D$, the same
holds for $\psi=\log Z$.
\end{proof}

\begin{rem}
No interchange of an uncontrolled infinite-dimensional limit with
differentiation is used here. The derivatives are justified directly
by finite-order Taylor estimates and the local exponential bound.
\end{rem}

\subsection{Covariance Hessian and strict convexity}
\label{subsec:covariance-hessian}

For
\[
\theta=(\beta,\lambda^1,\ldots,\lambda^N)\in\mathcal D,
\]
write
\[
F_\theta
=
-\beta H-\sum_{a=1}^{N}\lambda^a\Phi_a,
\qquad
\psi(\theta)
=
\log Z(\theta).
\]
For a parameter direction
\[
\dot\theta
=
(\dot\beta,\dot\lambda^1,\ldots,\dot\lambda^N),
\]
set
\[
\mathcal O_{\dot\theta}
=
\dot\beta H+\sum_{a=1}^{N}\dot\lambda^a\Phi_a.
\]
Then
\[
F_{\theta+t\dot\theta}
=
F_\theta-t\mathcal O_{\dot\theta}.
\]

Assume that the local exponential regularity hypothesis of
Definition~\ref{def:local-exponential-regularity} holds. Differentiation
under the normalized mean gives
\[
D\psi(\theta)[\dot\theta]
=
-\mathsf n_\theta^*(\mathcal O_{\dot\theta}).
\]
For two parameter directions $\dot\theta_1$ and $\dot\theta_2$, a second
differentiation gives
\[
D^2\psi(\theta)
[\dot\theta_1,\dot\theta_2]
=
\operatorname{Cov}_{\mathsf n_\theta^*}
\left(
\mathcal O_{\dot\theta_1},
\mathcal O_{\dot\theta_2}
\right),
\]
where
\[
\operatorname{Cov}_{\mathsf n}(f,g)
=
\mathsf n(fg)-\mathsf n(f)\mathsf n(g).
\]
In particular,
\[
D^2\psi(\theta)[\dot\theta,\dot\theta]
=
\operatorname{Var}_{\mathsf n_\theta^*}
(\mathcal O_{\dot\theta})
\geq 0,
\]
where
\[
\operatorname{Var}_{\mathsf n}(f)
=
\mathsf n(f^2)-\mathsf n(f)^2.
\]

\begin{Definition}[Thermodynamically null direction]
\label{def:null-direction}
A parameter direction $\dot\theta$ is called
\emph{thermodynamically null at $\theta$} if
\[
\operatorname{Var}_{\mathsf n_\theta^*}
(\mathcal O_{\dot\theta})=0.
\]
\end{Definition}

Every direction for which $\mathcal O_{\dot\theta}$ is constant is
thermodynamically null. If the equilibrium mean is not faithful, a
nonconstant observable may also have zero variance.

\begin{Proposition}[Strict-convexity criterion]
\label{prop:strict-convexity}
Let $U\subset\mathcal D$ be convex, and assume that $\psi$ is twice
differentiable on $U$. Then $\psi$ is convex on $U$.

If
\[
\operatorname{Var}_{\mathsf n_\theta^*}
(\mathcal O_{\dot\theta})>0
\]
for every $\theta\in U$ and every nonzero parameter direction
$\dot\theta$, then $\psi$ is strictly convex on $U$.

At every $\theta\in U$, the Hessian induces a positive-definite
quadratic form on
\[
T_\theta U\big/\ker D^2\psi(\theta).
\]
If these kernels have locally constant rank, the resulting quotient
spaces form a smooth vector bundle equipped with the metric induced by
the covariance form.
\end{Proposition}

\begin{proof}
The Hessian of $\psi$ is the covariance form. Hence
\[
D^2\psi(\theta)[\dot\theta,\dot\theta]
=
\operatorname{Var}_{\mathsf n_\theta^*}
(\mathcal O_{\dot\theta})
\geq 0,
\]
which proves convexity. Under the stated non-degeneracy assumption, the
restriction of $\psi$ to every nonconstant affine line contained in
$U$ has strictly positive second derivative and is therefore strictly
convex.

At a fixed $\theta$, the kernel of the positive semidefinite covariance
form is precisely $\ker D^2\psi(\theta)$. The induced quadratic form on
the quotient by this kernel is positive definite. If the kernels have
locally constant rank, they form a smooth subbundle, and the quotient
bundle inherits the corresponding positive-definite metric.
\end{proof}

\begin{rem}
The map
\[
f\longmapsto\log\mathsf m_0(e^f)
\]
cannot be strictly convex in the direction of constant functions,
because
\[
\log\mathsf m_0(e^{f+c})
=
c+\log\mathsf m_0(e^f).
\]
Strict convexity must therefore be understood modulo constants, or more
generally modulo covariance-null directions that persist along the
parameter domain.
\end{rem}

\begin{rem}
If $\mathsf n_\theta^*$ is faithful and
\[
\operatorname{Var}_{\mathsf n_\theta^*}(f)=0,
\]
then
\[
\mathsf n_\theta^*
\left(
\bigl(f-\mathsf n_\theta^*(f)\bigr)^2
\right)
=
0
\]
implies
\[
f=\mathsf n_\theta^*(f)
\]
in $\mathcal L$. In that case, the thermodynamically null directions
are exactly those for which
\[
\dot\beta H+\sum_{a=1}^{N}\dot\lambda^a\Phi_a
\]
is constant.
\end{rem}

\subsection{Legendre--Fenchel duality and thermodynamic entropy}
\label{subsec:legendre-duality}

Let
\[
\theta=(\beta,\lambda^1,\ldots,\lambda^N)
\]
and let
\[
x=(E,c_1,\ldots,c_N)
\]
denote the corresponding vector of extensive variables. We use the
pairing
\[
\langle\theta,x\rangle
=
\beta E+\sum_{a=1}^{N}\lambda^a c_a.
\]

The logarithmic partition functional is
\[
\psi(\theta)
=
\log\mathsf m_0
\left(
e^{-\beta H-\sum_{a=1}^{N}\lambda^a\Phi_a}
\right).
\]
Under the local exponential regularity hypothesis,
\[
\partial_\beta\psi(\theta)
=
-\mathsf n_\theta^*(H)
\]
and
\[
\partial_{\lambda^a}\psi(\theta)
=
-\mathsf n_\theta^*(\Phi_a).
\]
Therefore,
\[
\nabla\psi(\theta)
=
-x(\theta),
\]
where
\[
x(\theta)
=
\left(
\mathsf n_\theta^*(H),
\mathsf n_\theta^*(\Phi_1),
\ldots,
\mathsf n_\theta^*(\Phi_N)
\right).
\]

\begin{Definition}[Thermodynamic entropy]
\label{def:thermodynamic-entropy}
The thermodynamic entropy associated with $\psi$ is defined by
\[
\mathcal S(x)
=
\inf_{\theta\in\mathcal D}
\left\{
\psi(\theta)+\langle\theta,x\rangle
\right\}.
\]
Equivalently,
\[
\mathcal S(E,c)
=
\inf_{(\beta,\lambda)\in\mathcal D}
\left\{
\log Z(\beta,\lambda)
+
\beta E
+
\sum_{a=1}^{N}\lambda^a c_a
\right\}.
\]
\end{Definition}

\begin{Proposition}[Concavity and Legendre correspondence]
\label{prop:legendre}
Assume that $\mathcal D$ is convex and that $\psi$ is differentiable
and convex on $\mathcal D$. Then $\mathcal S$ is concave on its
effective domain.

For $\theta\in\mathcal D$, set
\[
x(\theta)=-\nabla\psi(\theta).
\]
Then $\theta$ realizes the infimum in
Definition~\ref{def:thermodynamic-entropy}, and
\[
\mathcal S(x(\theta))
=
\psi(\theta)+\langle\theta,x(\theta)\rangle.
\]

If $\psi$ descends to a strictly convex differentiable function on a
quotient of the convex parameter domain by a fixed space of
thermodynamically null directions, then the map
\[
\theta\longmapsto x(\theta)
\]
is injective on that quotient.
\end{Proposition}

\begin{proof}
For every fixed $\theta\in\mathcal D$, the function
\[
x\longmapsto
\psi(\theta)+\langle\theta,x\rangle
\]
is affine. Since an infimum of affine functions is concave,
$\mathcal S$ is concave.

Let
\[
x(\theta)=-\nabla\psi(\theta).
\]
For every $\theta'\in\mathcal D$, convexity of $\psi$ gives
\[
\psi(\theta')
\geq
\psi(\theta)
+
\left\langle
\nabla\psi(\theta),\theta'-\theta
\right\rangle.
\]
Using
\[
\nabla\psi(\theta)=-x(\theta),
\]
we obtain
\[
\psi(\theta')
+
\langle\theta',x(\theta)\rangle
\geq
\psi(\theta)
+
\langle\theta,x(\theta)\rangle.
\]
Thus $\theta$ realizes the infimum defining
$\mathcal S(x(\theta))$.

Finally, the gradient of a differentiable strictly convex function on a
convex domain is injective. Applying this statement to the induced
function on the quotient proves the last assertion.
\end{proof}

\begin{Proposition}[Constrained variational entropy]
\label{prop:constrained-entropy}
For an extensive vector
\[
x=(E,c_1,\ldots,c_N),
\]
define
\[
\mathcal S_{\mathrm{var}}(x)
=
\sup
\left\{
-\mathcal H(\mathsf n\,\|\,\mathsf m_0)
\;\middle|\;
\mathsf n\in\mathsf S(\mathcal L),
\quad
\mathsf n(H)=E,
\quad
\mathsf n(\Phi_a)=c_a
\right\}.
\]
Then
\[
\mathcal S_{\mathrm{var}}(x)
\leq
\mathcal S(x).
\]

Suppose that the hypotheses of
Proposition~\ref{prop:gibbs-variational-principle} hold. If
$x=x(\theta)$ for some $\theta\in\mathcal D$, then equality holds and
the supremum is attained at $\mathsf n_\theta^*$:
\[
\mathcal S_{\mathrm{var}}(x(\theta))
=
\mathcal S(x(\theta))
=
-\mathcal H(\mathsf n_\theta^*\,\|\,\mathsf m_0).
\]
\end{Proposition}

\begin{proof}
Let $\mathsf n$ satisfy the constraints associated with $x$. The
definition of relative entropy, applied to $F_\theta$, gives, for every
$\theta\in\mathcal D$,
\[
\mathcal H(\mathsf n\,\|\,\mathsf m_0)
\geq
\mathsf n(F_\theta)-\psi(\theta).
\]
Since
\[
\mathsf n(F_\theta)
=
-\langle\theta,x\rangle,
\]
we obtain
\[
-\mathcal H(\mathsf n\,\|\,\mathsf m_0)
\leq
\psi(\theta)+\langle\theta,x\rangle.
\]
Taking first the supremum over all states satisfying the constraints
and then the infimum over $\theta$ gives
\[
\mathcal S_{\mathrm{var}}(x)
\leq
\mathcal S(x).
\]

If $x=x(\theta)$, the equilibrium mean $\mathsf n_\theta^*$ satisfies
the required constraints. By the change-of-reference identity,
evaluated at $\mathsf n_\theta^*$,
\[
\mathcal H(\mathsf n_\theta^*\,\|\,\mathsf m_0)
=
\mathsf n_\theta^*(F_\theta)-\psi(\theta).
\]
Since
\[
\mathsf n_\theta^*(F_\theta)
=
-\langle\theta,x(\theta)\rangle,
\]
we obtain
\[
-\mathcal H(\mathsf n_\theta^*\,\|\,\mathsf m_0)
=
\psi(\theta)+\langle\theta,x(\theta)\rangle
=
\mathcal S(x(\theta)).
\]
This proves equality.
\end{proof}

\begin{rem}
Equality between $\mathcal S_{\mathrm{var}}$ and $\mathcal S$ for every
admissible extensive vector requires an additional duality or
constraint-qualification hypothesis. Without such a hypothesis, the
equality is guaranteed on the equilibrium image
\[
-\nabla\psi(\mathcal D)
\]
but need not hold on the entire boundary of the constraint domain.
\end{rem}
\section{Existence and Uniqueness of Equilibrium Means on Fréchet Manifolds}

\subsection{Analytic hypotheses and functional framework}

Let $(\mathcal M,\omega)$ be a weak symplectic Fréchet manifold, let
$\mathcal L$ be the unital vector lattice algebra of observables
introduced in the preceding section, and let $\mathsf m_0$ be a
normalized mean on $\mathcal L$. We fix a possibly unbounded Hamiltonian
$H\in\mathcal L$ and a finite family of possibly unbounded observables
\[
\Phi_1,\ldots,\Phi_N\in\mathcal L,
\]
for instance components of a moment map.

For
\[
(\beta,\lambda)\in\mathbb R\times\mathbb R^N,
\]
set
\[
F_{\beta,\lambda}
=
-\beta H-\sum_{a=1}^{N}\lambda^a\Phi_a.
\]
The admissible thermodynamic domain is
\[
\mathcal D
=
\left\{
(\beta,\lambda)\in(0,\infty)\times\mathbb R^N
\;\middle|\;
F_{\beta,\lambda}\in\mathcal E(\mathsf m_0)
\right\}.
\]

For $(\beta,\lambda)\in\mathcal D$, the central functional is the
\emph{free energy}
\[
\mathcal G_{\beta,\lambda}(\mathsf n)
=
\mathcal H(\mathsf n\,\|\,\mathsf m_0)
+
\beta\,\mathsf n(H)
+
\sum_{a=1}^{N}\lambda^a\mathsf n(\Phi_a),
\]
defined on $\mathsf S(\mathcal L)$. Equivalently,
\[
\mathcal G_{\beta,\lambda}(\mathsf n)
=
\mathcal H(\mathsf n\,\|\,\mathsf m_0)
-
\mathsf n(F_{\beta,\lambda}).
\]

\begin{Definition}[Coercivity and local exponential regularity]
\label{def:coercive}
We say that $(H,\Phi_a,\mathsf m_0)$ satisfies the analytic hypotheses
on a parameter domain if:
\begin{enumerate}
\item[\textup{(i)}] there exists a continuous proper functional
\[
V\colon\mathcal M\longrightarrow[0,\infty)
\]
and constants $c_1,c_2>0$ such that
\[
H(m)\geq c_1V(m)-c_2,
\qquad
m\in\mathcal M;
\]
\item[\textup{(ii)}] the partition functional is locally exponentially
regular on $\mathcal D$ in the sense of
Definition~\ref{def:local-exponential-regularity};
\item[\textup{(iii)}] whenever a compactness argument is used, the
sublevel sets of $\mathcal G_{\beta,\lambda}$ are compact in
$\mathsf S(\mathcal L)$ for the topology of pointwise convergence.
\end{enumerate}
\end{Definition}

\begin{rem}
The lower bound in Definition~\ref{def:coercive}\textup{(i)} is a
model-space coercivity condition. For general finitely additive means,
it does not by itself imply compactness of free-energy sublevel sets.
Such compactness must be verified separately, as required in
condition~\textup{(iii)}.
\end{rem}

\subsection{Admissible parameters and non-vanishing partition functions}
\label{subsec:admissible-partition}

For a thermodynamic potential
\[
F_{\beta,\lambda}
=
-\beta H-\sum_{a=1}^{N}\lambda^a\Phi_a,
\]
the positivity of the function $e^{F_{\beta,\lambda}}$ does not, for an
arbitrary normalized mean, automatically imply
\[
\mathsf m_0(e^{F_{\beta,\lambda}})>0.
\]
A positive normalized linear functional need not be faithful and may
vanish on a nonzero positive function.

This issue is particularly important for normalized means defined by
finite-dimensional cut-offs. Let
\[
Z_K(\beta,\lambda)
=
\frac{
\displaystyle
\int_{B_K}
e^{F_{\beta,\lambda}}\,\mathrm d\mu_K
}{
\displaystyle
\int_{B_K}\mathrm d\mu_K
}.
\]
Although
\[
Z_K(\beta,\lambda)>0
\qquad
\text{for every }K,
\]
the limit of the sequence may vanish. Positivity of all the
finite-dimensional partition functionals therefore does not by itself
imply strict positivity of the limiting partition functional.

\begin{Definition}[Admissible thermodynamic domain]
\label{def:admissible-domain}
The admissible thermodynamic domain is
\[
\mathcal D
=
\left\{
(\beta,\lambda)\in(0,\infty)\times\mathbb R^N
\;\middle|\;
F_{\beta,\lambda}\in\mathcal E(\mathsf m_0)
\right\}.
\]
Equivalently, $(\beta,\lambda)\in\mathcal D$ if
\[
e^{F_{\beta,\lambda}}\in\mathcal L,
\qquad
0<
\mathsf m_0(e^{F_{\beta,\lambda}})
<\infty,
\]
and
\[
fe^{F_{\beta,\lambda}}\in\mathcal L
\qquad
\text{for every }f\in\mathcal L.
\]
\end{Definition}

Thus strict positivity of the limiting partition functional is part of
the admissibility condition unless it follows from an independent
structural assumption, such as faithfulness of the reference mean.

\begin{Theorem}[Criteria for finiteness and non-vanishing]
\label{thm:partition-positive}
\label{prop:partition-nonvanishing}
Let $(\beta,\lambda)$ be such that
\[
e^{F_{\beta,\lambda}}\in\mathcal L
\]
and
\[
\mathsf m_0(e^{F_{\beta,\lambda}})<\infty.
\]
Then
\[
0<
\mathsf m_0(e^{F_{\beta,\lambda}})
<\infty
\]
under either of the following assumptions:
\begin{enumerate}
\item[\textup{(i)}] the reference mean $\mathsf m_0$ is faithful on
$\mathcal L$;
\item[\textup{(ii)}] the reference mean is defined by cut-offs, one has
\[
\mathsf m_0(e^{F_{\beta,\lambda}})
=
\lim_{K\to\infty}Z_K(\beta,\lambda),
\]
and
\[
\liminf_{K\to\infty}Z_K(\beta,\lambda)>0.
\]
\end{enumerate}
\end{Theorem}

\begin{proof}
Finiteness is assumed. Since
\[
e^{F_{\beta,\lambda}}>0
\]
pointwise, it is a nonzero positive element of $\mathcal L$. Under
assumption~\textup{(i)}, faithfulness gives
\[
\mathsf m_0(e^{F_{\beta,\lambda}})>0.
\]

Under assumption~\textup{(ii)},
\[
\mathsf m_0(e^{F_{\beta,\lambda}})
=
\lim_{K\to\infty}Z_K(\beta,\lambda)
\geq
\liminf_{K\to\infty}Z_K(\beta,\lambda)>0.
\]
\end{proof}

\begin{rem}
The lower-bound condition in
Theorem~\ref{thm:partition-positive} is sufficient but not necessary. In
applications, it may be replaced by any estimate that directly yields
\[
\mathsf m_0(e^{F_{\beta,\lambda}})>0.
\]
\end{rem}

\subsection{Change of reference under exponential tilting}
\label{subsec:change-reference}

Let $F\in\mathcal E(\mathsf m_0)$, and let
\[
Z_F=\mathsf m_0(e^F).
\]
The exponentially tilted mean
\[
\mathsf m_F(f)
=
\frac{\mathsf m_0(fe^F)}{Z_F},
\qquad
f\in\mathcal L,
\]
is then well defined on $\mathcal L$.

\begin{Definition}[Translation compatibility]
\label{def:translation-compatibility}
The potential $F\in\mathcal E(\mathsf m_0)$ is said to be
\emph{translation compatible with $\mathsf m_0$} if
\[
\mathcal E(\mathsf m_F)+F
=
\mathcal E(\mathsf m_0).
\]
Equivalently, the affine map
\[
\tau_F\colon
\mathcal E(\mathsf m_F)
\longrightarrow
\mathcal E(\mathsf m_0),
\qquad
\tau_F(g)=g+F,
\]
is a bijection.
\end{Definition}

\begin{Theorem}[Change-of-reference identity]
\label{thm:change-reference}
Let $F\in\mathcal E(\mathsf m_0)$ be translation compatible with
$\mathsf m_0$. Then, for every
$\mathsf n\in\mathsf S(\mathcal L)$,
\[
\boxed{
\mathcal H(\mathsf n\,\|\,\mathsf m_F)
=
\mathcal H(\mathsf n\,\|\,\mathsf m_0)
-
\mathsf n(F)
+
\log Z_F.
}
\]
The identity is understood in $(-\infty,+\infty]$. Since
$\mathsf n(F)$ and $\log Z_F$ are finite, no indeterminate expression
occurs.
\end{Theorem}

\begin{proof}
For every
\[
g\in\mathcal E(\mathsf m_F),
\]
the definition of the tilted mean gives
\[
\mathsf m_F(e^g)
=
\frac{
\mathsf m_0(e^{g+F})
}{
Z_F
}.
\]
Consequently,
\begin{align*}
\mathcal H(\mathsf n\,\|\,\mathsf m_F)
&=
\sup_{g\in\mathcal E(\mathsf m_F)}
\left\{
\mathsf n(g)-\log\mathsf m_F(e^g)
\right\}
\\
&=
\sup_{g\in\mathcal E(\mathsf m_F)}
\left\{
\mathsf n(g)
-
\log\mathsf m_0(e^{g+F})
+
\log Z_F
\right\}.
\end{align*}
Set
\[
h=g+F.
\]
By translation compatibility, the map
\[
g\longmapsto h=g+F
\]
is a bijection from $\mathcal E(\mathsf m_F)$ onto
$\mathcal E(\mathsf m_0)$. Moreover,
\[
\mathsf n(g)
=
\mathsf n(h-F)
=
\mathsf n(h)-\mathsf n(F).
\]
Therefore
\begin{align*}
\mathcal H(\mathsf n\,\|\,\mathsf m_F)
&=
\sup_{h\in\mathcal E(\mathsf m_0)}
\left\{
\mathsf n(h)
-
\mathsf n(F)
-
\log\mathsf m_0(e^h)
+
\log Z_F
\right\}
\\
&=
\mathcal H(\mathsf n\,\|\,\mathsf m_0)
-
\mathsf n(F)
+
\log Z_F.
\end{align*}
\end{proof}

\begin{Proposition}[Free-energy decomposition]
\label{cor:free-energy-decomposition}
Let
\[
F_{\beta,\lambda}
=
-\beta H-\sum_{a=1}^{N}\lambda^a\Phi_a
\]
be translation compatible with $\mathsf m_0$, and let
\[
\mathsf n_{\beta,\lambda}^*
=
\mathsf m_{F_{\beta,\lambda}}.
\]
Then
\[
\boxed{
\mathcal G_{\beta,\lambda}(\mathsf n)
=
\mathcal H(
\mathsf n\,\|\,\mathsf n_{\beta,\lambda}^*
)
-
\log Z(\beta,\lambda).
}
\]
\end{Proposition}

\begin{proof}
Theorem~\ref{thm:change-reference} gives
\[
\mathcal H(
\mathsf n\,\|\,\mathsf n_{\beta,\lambda}^*
)
=
\mathcal H(\mathsf n\,\|\,\mathsf m_0)
-
\mathsf n(F_{\beta,\lambda})
+
\log Z(\beta,\lambda).
\]
Since
\[
\mathcal G_{\beta,\lambda}(\mathsf n)
=
\mathcal H(\mathsf n\,\|\,\mathsf m_0)
-
\mathsf n(F_{\beta,\lambda}),
\]
the result follows.
\end{proof}

\begin{Proposition}[Existence and uniqueness of the minimizer]
\label{cor:change-reference-minimizer}
Assume that $F_{\beta,\lambda}$ is translation compatible with
$\mathsf m_0$ and that the relative entropy with reference
$\mathsf n_{\beta,\lambda}^*$ satisfies the separation and local
perturbation hypotheses of
Proposition~\ref{prop:entropy-properties}. Then
\[
\mathcal G_{\beta,\lambda}(\mathsf n)
\geq
-\log Z(\beta,\lambda),
\]
with equality if and only if
\[
\mathsf n=\mathsf n_{\beta,\lambda}^*.
\]
\end{Proposition}

\begin{proof}
By Proposition~\ref{cor:free-energy-decomposition},
\[
\mathcal G_{\beta,\lambda}(\mathsf n)
+
\log Z(\beta,\lambda)
=
\mathcal H(
\mathsf n\,\|\,\mathsf n_{\beta,\lambda}^*
)
\geq 0.
\]
Equality holds if and only if
\[
\mathcal H(
\mathsf n\,\|\,\mathsf n_{\beta,\lambda}^*
)=0.
\]
The separation property of
Proposition~\ref{prop:entropy-properties}, applied with reference
$\mathsf n_{\beta,\lambda}^*$, then implies
\[
\mathsf n=\mathsf n_{\beta,\lambda}^*.
\]
\end{proof}

\subsection{Existence theorem for equilibrium means}

For the sake of clarity, we collect the preceding statements in a
single theorem.

\begin{Theorem}[Existence and uniqueness of equilibrium means]
\label{thm:existence}
Let $(\beta,\lambda)\in\mathcal D$. Assume that
$F_{\beta,\lambda}$ is translation compatible with $\mathsf m_0$ and
that the relative entropy with reference
$\mathsf n_{\beta,\lambda}^*$ satisfies the separation and local
perturbation hypotheses of
Proposition~\ref{prop:entropy-properties}. Then there exists a unique
normalized mean
\[
\mathsf n_{\beta,\lambda}^{*}\in\mathsf S(\mathcal L)
\]
minimizing $\mathcal G_{\beta,\lambda}$. It is explicitly given by
\[
\boxed{
\mathsf n_{\beta,\lambda}^{*}(f)
=
\frac{
\mathsf m_0\left(
f\,e^{-\beta H-\sum_{a=1}^{N}\lambda^a\Phi_a}
\right)
}{
\mathsf m_0\left(
e^{-\beta H-\sum_{a=1}^{N}\lambda^a\Phi_a}
\right)
},
\qquad
f\in\mathcal L.
}
\]
Moreover,
\[
\inf_{\mathsf n\in\mathsf S(\mathcal L)}
\mathcal G_{\beta,\lambda}(\mathsf n)
=
-\log Z(\beta,\lambda).
\]

If the local exponential regularity hypothesis of
Definition~\ref{def:local-exponential-regularity} holds, then
\[
\psi(\beta,\lambda)
=
\log Z(\beta,\lambda)
\]
is smooth on $\mathcal D$. Its Hessian is the covariance form described
in Subsection~\ref{subsec:covariance-hessian}. Consequently, $\psi$ is
convex on every convex subset of $\mathcal D$.

If the covariance form is positive definite in every nonzero parameter
direction on a convex subset $U\subset\mathcal D$, then $\psi$ is
strictly convex on $U$. More generally, at each
$\theta\in\mathcal D$, its Hessian induces a positive-definite form on
\[
T_\theta\mathcal D\big/\ker D^2\psi(\theta).
\]

The thermodynamic entropy is the concave Legendre--Fenchel dual
\[
\mathcal S(E,c)
=
\inf_{(\beta,\lambda)\in\mathcal D}
\left\{
\log Z(\beta,\lambda)
+
\beta E
+
\sum_{a=1}^{N}\lambda^a c_a
\right\}.
\]
It is concave on its effective domain. If $\mathcal D$ is convex, the
Legendre correspondence and its relation with constrained variational
entropy are those described in
Propositions~\ref{prop:legendre} and~\ref{prop:constrained-entropy}.
\end{Theorem}

\begin{proof}
Existence, uniqueness, and the value of the minimum follow from
Proposition~\ref{cor:change-reference-minimizer}. Under local
exponential regularity,
Proposition~\ref{prop:partition-smoothness} gives
\[
\partial_\beta\log Z(\beta,\lambda)
=
-\mathsf n_{\beta,\lambda}^{*}(H)
\]
and
\[
\partial_{\lambda^a}\log Z(\beta,\lambda)
=
-\mathsf n_{\beta,\lambda}^{*}(\Phi_a).
\]
A second differentiation gives
\[
D^2\log Z(\beta,\lambda)
[\dot\theta_1,\dot\theta_2]
=
\operatorname{Cov}_{\mathsf n_{\beta,\lambda}^{*}}
\left(
\mathcal O_{\dot\theta_1},
\mathcal O_{\dot\theta_2}
\right).
\]
Convexity and the strict-convexity criterion now follow from
Proposition~\ref{prop:strict-convexity}.

Finally, $\mathcal S$ is an infimum of affine functions of $(E,c)$ and
is therefore concave.
\end{proof}

\begin{rem}
The compactness assumption in
Definition~\ref{def:coercive}\textup{(iii)} is not needed in
Theorem~\ref{thm:existence}, because the minimizer is explicitly
constructed by exponential tilting. It becomes relevant in variational
existence arguments where no compatible exponential tilt is known in
advance.
\end{rem}

\subsection{Invariance under Hamiltonian flows}

\begin{Theorem}[Stationarity of equilibrium means]
\label{thm:stationarity}
Assume that the Hamiltonian vector field $X_H$ generates a global flow
$\Phi_t^H$ on $\mathcal M$ preserving $\mathsf m_0$ and $\mathcal L$.
Assume moreover that
\[
H\circ\Phi_t^H=H
\]
and
\[
\Phi_a\circ\Phi_t^H=\Phi_a,
\qquad
a=1,\ldots,N.
\]
Then every equilibrium mean of Theorem~\ref{thm:existence} satisfies
\[
\mathsf n_{\beta,\lambda}^{*}(f\circ\Phi_t^H)
=
\mathsf n_{\beta,\lambda}^{*}(f)
\]
for all $f\in\mathcal L$ and $t\in\mathbb R$ for which
$f\circ\Phi_t^H\in\mathcal L$.
\end{Theorem}

\begin{proof}
Since the Hamiltonian and the constraint observables are invariant,
\[
F_{\beta,\lambda}\circ\Phi_t^H
=
F_{\beta,\lambda}.
\]
Consequently,
\begin{align*}
\mathsf n_{\beta,\lambda}^{*}(f\circ\Phi_t^H)
&=
\frac{
\mathsf m_0\left(
(f\circ\Phi_t^H)e^{F_{\beta,\lambda}}
\right)
}{
Z(\beta,\lambda)
}
\\
&=
\frac{
\mathsf m_0\left(
(fe^{F_{\beta,\lambda}})\circ\Phi_t^H
\right)
}{
Z(\beta,\lambda)
}
\\
&=
\frac{
\mathsf m_0\left(
fe^{F_{\beta,\lambda}}
\right)
}{
Z(\beta,\lambda)
}
\\
&=
\mathsf n_{\beta,\lambda}^{*}(f).
\end{align*}
\end{proof}

\begin{rem}
For an autonomous Hamiltonian flow, the identity
\[
H\circ\Phi_t^H=H
\]
holds whenever the flow and the Hamiltonian calculus are defined. It is
displayed explicitly above to emphasize that the entire exponential
weight must be preserved.
\end{rem}

\begin{rem}
Stationarity alone does not imply a Kubo--Martin--Schwinger boundary
condition. A KMS-type identity requires an additional analytic or
Poisson-algebraic structure and must be proved separately.
\end{rem}

\subsection{A classical Poisson--KMS identity}
\label{subsec:poisson-kms}

Stationarity of an equilibrium mean does not by itself imply the
operator-algebraic Kubo--Martin--Schwinger boundary condition. In the
present commutative Hamiltonian setting, the appropriate analogue is
instead formulated in terms of the Poisson bracket.

Let $\mathcal P\subset\mathcal L$ be a unital Poisson algebra of smooth
observables. We use the convention
\[
\{f,g\}=X_f(g).
\]

\begin{Definition}[Poisson-invariant reference mean]
\label{def:poisson-invariant-mean}
A normalized mean $\mathsf m_0$ on $\mathcal L$ is called
\emph{Poisson invariant on $\mathcal P$} if
\[
\mathsf m_0(\{f,g\})=0
\]
for every $f,g\in\mathcal P$ for which the bracket belongs to
$\mathcal L$.
\end{Definition}

This condition is the algebraic counterpart of integration by parts
with respect to a Liouville-type reference measure.

\begin{Definition}[Classical Poisson--KMS state]
\label{def:poisson-kms}
Let $H\in\mathcal P$ and let $\beta>0$. A normalized mean $\mathsf n$ is
called a \emph{classical Poisson--KMS state at inverse temperature
$\beta$} if
\[
\mathsf n(\{f,g\})
=
\beta\,\mathsf n\bigl(g\{f,H\}\bigr)
\]
for every pair of observables $f,g\in\mathcal P$ for which all the terms
are defined.
\end{Definition}

In the presence of additional conserved constraints, the corresponding
identity contains their Hamiltonian generators.

\begin{Theorem}[Poisson--KMS identity for equilibrium means]
\label{thm:poisson-kms}
Assume that $\mathsf m_0$ is Poisson invariant on $\mathcal P$. Let
\[
F_{\beta,\lambda}
=
-\beta H-\sum_{a=1}^{N}\lambda^a\Phi_a
\]
be an admissible potential, where
\[
H,\Phi_1,\ldots,\Phi_N\in\mathcal P.
\]
Assume that, for every pair $f,g\in\mathcal P$ under consideration,
\[
ge^{F_{\beta,\lambda}}\in\mathcal P,
\]
that all products and brackets occurring below belong to $\mathcal L$,
and that the usual Leibniz and chain rules are valid. Then the
equilibrium mean $\mathsf n_{\beta,\lambda}^*$ satisfies
\[
\mathsf n_{\beta,\lambda}^*(\{f,g\})
=
\beta\,
\mathsf n_{\beta,\lambda}^*
\bigl(g\{f,H\}\bigr)
+
\sum_{a=1}^{N}\lambda^a
\mathsf n_{\beta,\lambda}^*
\bigl(g\{f,\Phi_a\}\bigr).
\]

In particular, in the absence of additional constraints,
\[
\mathsf n_{\beta}^*(\{f,g\})
=
\beta\,
\mathsf n_{\beta}^*
\bigl(g\{f,H\}\bigr).
\]
Thus $\mathsf n_\beta^*$ is a classical Poisson--KMS state at inverse
temperature $\beta$.
\end{Theorem}

\begin{proof}
Since
\[
ge^{F_{\beta,\lambda}}\in\mathcal P,
\]
Poisson invariance of $\mathsf m_0$ gives
\[
\mathsf m_0\left(
\{f,ge^{F_{\beta,\lambda}}\}
\right)
=
0.
\]
Using the Leibniz rule,
\[
\{f,ge^{F_{\beta,\lambda}}\}
=
\{f,g\}e^{F_{\beta,\lambda}}
+
g\{f,e^{F_{\beta,\lambda}}\}.
\]
The chain rule gives
\[
\{f,e^{F_{\beta,\lambda}}\}
=
e^{F_{\beta,\lambda}}
\{f,F_{\beta,\lambda}\}.
\]
Since
\[
\{f,F_{\beta,\lambda}\}
=
-\beta\{f,H\}
-
\sum_{a=1}^{N}\lambda^a\{f,\Phi_a\},
\]
we obtain
\begin{align*}
\mathsf m_0\left(
\{f,g\}e^{F_{\beta,\lambda}}
\right)
&=
-\mathsf m_0\left(
g e^{F_{\beta,\lambda}}
\{f,F_{\beta,\lambda}\}
\right)
\\
&=
\beta\,
\mathsf m_0\left(
g\{f,H\}e^{F_{\beta,\lambda}}
\right)
\\
&\quad+
\sum_{a=1}^{N}\lambda^a
\mathsf m_0\left(
g\{f,\Phi_a\}e^{F_{\beta,\lambda}}
\right).
\end{align*}
Dividing by $Z(\beta,\lambda)$ yields the stated identity.
\end{proof}

\begin{rem}
The identity in Theorem~\ref{thm:poisson-kms} is a classical
Poisson-algebraic equilibrium condition. It must not be confused with
the complex-time KMS boundary condition for noncommutative
$C^*$-dynamical systems.
\end{rem}

\begin{rem}
If the constraint observables are Casimirs, then
\[
\{f,\Phi_a\}=0,
\qquad
a=1,\ldots,N,
\]
and the generalized identity reduces to
\[
\mathsf n_{\beta,\lambda}^*(\{f,g\})
=
\beta\,
\mathsf n_{\beta,\lambda}^*
\bigl(g\{f,H\}\bigr).
\]
\end{rem}

\subsection{Compactness and variational limits}

Suppose, in addition, that $\mathcal L$ is a unital normed order space
for which every normalized positive functional is continuous and has
norm one. Then $\mathsf S(\mathcal L)$ is contained in the continuous
dual $\mathcal L^*$ and may be identified with a weak-$*$ closed subset
of its unit ball. It is therefore weak-$*$ compact by the
Banach--Alaoglu theorem.

Since
\[
\mathsf n\longmapsto
\mathcal H(\mathsf n\,\|\,\mathsf m_0)
\]
is weak-$*$ lower semicontinuous, the same is true of
$\mathcal G_{\beta,\lambda}$: indeed, for
$H,\Phi_a\in\mathcal L$, the maps
\[
\mathsf n\longmapsto\mathsf n(H),
\qquad
\mathsf n\longmapsto\mathsf n(\Phi_a)
\]
are weak-$*$ continuous.

More generally, in the algebraic setting it is sufficient to assume
directly that the sublevel sets of
$\mathcal G_{\beta,\lambda}$ are compact for the topology of pointwise
convergence. This is the compactness condition appearing in
Definition~\ref{def:coercive}.

\begin{rem}
Compactness of the state space is automatic in the standard unital
$C^*$-algebraic setting. For a general algebra containing unbounded
observables, it must instead be verified or included among the analytic
hypotheses.
\end{rem}

\subsection{The Hilbertian case as a particular instance}

When $\mathcal M$ is modelled on a separable real Hilbert space
$\mathscr H$ with scalar product $\langle\cdot,\cdot\rangle$, the
preceding constructions include the classical Gaussian framework.

\begin{Proposition}[Hilbertian specialization]
\label{prop:hilbertian}
Let $\mathcal M=\mathscr H$ and
\[
\omega(u,v)=\langle Ju,v\rangle,
\]
where
\[
J\colon\mathscr H\longrightarrow\mathscr H
\]
is bounded, skew-adjoint, and invertible. Let $\mu_0$ be a centered
Gaussian Radon probability measure on $\mathscr H$ with nonnegative
selfadjoint trace-class covariance operator $C$, and let
\[
\mathsf m_0(f)
=
\int_{\mathscr H}f\,\mathrm d\mu_0
\]
on the corresponding integrability domain. Then:
\begin{enumerate}
\item[\textup{(1)}] for every probability measure $\nu$ on
$\mathscr H$, the variational entropy coincides with the usual
Kullback--Leibler relative entropy,
\[
\mathcal H(\nu\,\|\,\mu_0)
=
\begin{cases}
\displaystyle
\int_{\mathscr H}
\log\left(\frac{\mathrm d\nu}{\mathrm d\mu_0}\right)
\,\mathrm d\nu,
&
\nu\ll\mu_0,
\\[2ex]
+\infty,
&
\nu\not\ll\mu_0,
\end{cases}
\]
provided the admissible potentials contain an entropy-determining
class, for instance $C_b(\mathscr H)$;

\item[\textup{(2)}] for every $(\beta,\lambda)\in\mathcal D$, the
equilibrium mean is represented by the Gibbs probability measure
\[
\mathrm d\nu_{\beta,\lambda}^{*}(u)
=
\frac{1}{Z(\beta,\lambda)}
e^{-\beta H(u)-\sum_{a=1}^{N}\lambda^a\Phi_a(u)}
\,\mathrm d\mu_0(u);
\]

\item[\textup{(3)}] whenever the translation-compatibility and
separation hypotheses of Theorem~\ref{thm:existence} hold, its
conclusions remain valid in this Gaussian setting. If, in addition, the
hypotheses of Theorem~\ref{thm:stationarity} are satisfied, then
$\nu_{\beta,\lambda}^{*}$ is invariant under the Hamiltonian flow.
\end{enumerate}
\end{Proposition}

\begin{proof}
The first assertion is the standard variational characterization of
relative entropy on a Polish space, applied to the entropy-determining
class contained in the admissible potentials. The second assertion
follows directly from the definition of exponential tilting. The third
is an immediate application of
Theorems~\ref{thm:existence} and~\ref{thm:stationarity}.
\end{proof}

\begin{rem}
The Hilbertian framework is a $\sigma$-additive instance of the general
construction. It provides an important consistency test, while the
normalized-mean framework is intended to include situations in which no
suitable $\sigma$-additive reference measure is available.
\end{rem}

\section{Applications: $H^s$-Geodesic Equations on
$\operatorname{Map}(M,G)$ and on $\operatorname{Diff}(M)$}

Throughout this section, $M$ is a compact smooth manifold without
boundary, endowed with a smooth volume form $\operatorname{vol}_M$.
When a Riemannian metric $g$ is used,
$\Delta_M=\Delta_g$ denotes the nonnegative Laplace--Beltrami operator,
acting componentwise on functions, sections, or vector fields. Let $G$
be a finite-dimensional Lie group with Lie algebra $\mathfrak g$,
equipped with an $\operatorname{Ad}$-invariant inner product
$\langle\cdot,\cdot\rangle_{\mathfrak g}$. We write $d=\dim M$ and
assume
\[
r>\frac d2+1
\]
when Hilbert completions of Sobolev order $r$ are used.

The Hamiltonian phase space of a geodesic equation is naturally a
cotangent bundle, or, after reduction, a Lie--Poisson space. Since a
general Lie--Poisson space need not be symplectic, the symplectic
statements of the preceding sections are applied either on the full
cotangent bundle or on a suitable coadjoint orbit. In the formulas
below, we use the corresponding reduced velocity and momentum
variables.

Cylindrical functionals are understood relative to a fixed family of
finite-dimensional maps, such as spectral projections. We denote by
\[
\mathcal A_b\subset C_b(\mathcal M)
\]
the unital algebra of bounded continuous cylindrical functionals and by
$\mathcal L$ a larger vector lattice algebra containing the possibly
unbounded Hamiltonians, moment-map components, and exponential weights
used below. All equilibrium statements in this section are conditional
on the admissibility, translation-compatibility, separation, and
regularity hypotheses of the abstract theory.

\subsection{$H^s$ right-invariant metrics on the current group
$\operatorname{Map}(M,G)$}

\subsubsection*{Geometry and geodesic equation}

Consider the Fréchet--Lie group
\[
\mathcal G=\operatorname{Map}(M,G)
\]
with pointwise multiplication. Its Lie algebra is
\[
\mathfrak X_G=C^\infty(M,\mathfrak g),
\]
with pointwise bracket
\[
[\xi,\eta](x)=[\xi(x),\eta(x)]_{\mathfrak g}.
\]
For $s\geq0$, define the elliptic positive selfadjoint inertia operator
\[
A_s=(1+\Delta_M)^s\otimes\operatorname{id}_{\mathfrak g}.
\]
The corresponding inner product at the identity is
\[
\langle\xi,\eta\rangle_{H^s}
=
\int_M
\left\langle
A_s^{1/2}\xi(x),A_s^{1/2}\eta(x)
\right\rangle_{\mathfrak g}
\,\operatorname{vol}_M(x).
\]
It defines a right-invariant weak Riemannian metric on $\mathcal G$.

Let $\xi(t,\cdot)\in\mathfrak X_G$ be the right-trivialized velocity and
let
\[
m=A_s\xi
\]
be the corresponding momentum variable, identified with a
$\mathfrak g^*$-valued density through the inner product and the volume
form. With the convention
\[
\langle\operatorname{ad}_\xi^*m,\eta\rangle
=
\langle m,[\eta,\xi]\rangle,
\]
the Euler--Poincaré equation is
\[
\partial_t m+\operatorname{ad}_\xi^*m=0,
\qquad
m=A_s\xi.
\]
The coadjoint operator acts pointwise:
\[
(\operatorname{ad}_\xi^*m)(x)
=
\operatorname{ad}_{\xi(x)}^*m(x).
\]
Thus the equation is a field of finite-dimensional Euler equations
coupled through the elliptic operator $A_s$.

\subsubsection*{Hamiltonian and momentum variables}

The reduced kinetic Hamiltonian is
\[
H(\xi)
=
\frac12
\int_M
\langle\xi(x),A_s\xi(x)\rangle_{\mathfrak g}
\,\operatorname{vol}_M(x).
\]
Equivalently, in the momentum variable,
\[
H(m)
=
\frac12
\int_M
\langle A_s^{-1}m(x),m(x)\rangle_{\mathfrak g}
\,\operatorname{vol}_M(x),
\]
whenever the expression is defined.

The cotangent lifts of left and right translations on $\mathcal G$
possess the standard momentum maps. Their expressions depend on the
chosen trivialization of $T^*\mathcal G$. In reduced variables, $m$ is
the basic momentum coordinate, while the corresponding unreduced
conserved momentum is obtained from $m$ by the appropriate coadjoint
transport. Spatial symmetries of $(M,\operatorname{vol}_M)$, for
instance a compact subgroup of volume-preserving diffeomorphisms
commuting with $A_s$, may produce additional conserved quantities.

\subsubsection*{Equilibrium means}

Let $\mathcal A_b$ consist of bounded continuous cylindrical
functionals of $\xi$, or equivalently of $m$ on the domain where
$A_s^{-1}$ is defined, and let $\mathcal L$ contain $\mathcal A_b$, the
Sobolev energy $H$, and the momentum constraints considered below.
Choose a reference normalized mean $\mathsf m_0$ on $\mathcal L$, for
example a Gaussian mean with covariance related to $A_s^{-1}$ when this
covariance defines a Gaussian Radon measure on the chosen Hilbert
completion, or a normalized cut-off mean in the sense of
Example~\ref{ex:cutoff-mean}.

For a test field
\[
\lambda\in C^\infty(M,\mathfrak g),
\]
define the linear momentum observable
\[
\langle\lambda,m\rangle
=
\int_M
\langle\lambda(x),m(x)\rangle_{\mathfrak g}
\,\operatorname{vol}_M(x).
\]
The Gibbs-type variational functional is
\[
\mathcal G_{\beta,\lambda}(\mathsf n)
=
\mathcal H(\mathsf n\,\|\,\mathsf m_0)
+
\beta\,\mathsf n(H)
+
\mathsf n\bigl(\langle\lambda,m\rangle\bigr).
\]
Whenever
\[
F_{\beta,\lambda}
=
-\beta H-\langle\lambda,m\rangle
\]
belongs to $\mathcal E(\mathsf m_0)$ and satisfies the
translation-compatibility and separation hypotheses of
Theorem~\ref{thm:existence}, the unique minimizer is
\[
\boxed{
\mathsf n_{\beta,\lambda}^*(f)
=
\frac{
\mathsf m_0
\left(
f e^{-\beta H-\langle\lambda,m\rangle}
\right)
}{
\mathsf m_0
\left(
e^{-\beta H-\langle\lambda,m\rangle}
\right)
},
\qquad
f\in\mathcal L.
}
\]

The corresponding partition functional
\[
Z(\beta,\lambda)
=
\mathsf m_0
\left(
e^{-\beta H-\langle\lambda,m\rangle}
\right)
\]
generates the thermodynamic duality described in
Proposition~\ref{prop:legendre}, whenever its hypotheses are satisfied.
If the geodesic flow is globally defined, preserves $\mathsf m_0$ and
$\mathcal L$, and leaves the momentum observable
$\langle\lambda,m\rangle$ invariant, then
$\mathsf n_{\beta,\lambda}^*$ is stationary by
Theorem~\ref{thm:stationarity}. If, in addition, the extended
Poisson-invariance assumptions of Theorem~\ref{thm:poisson-kms} hold,
then the equilibrium mean satisfies the corresponding classical
Poisson--KMS identity.

\begin{rem}[Structure of equilibria]
If $\lambda$ is independent of $x$ and belongs to the centre of
$\mathfrak g$, then $\langle\lambda,m\rangle$ is invariant under the
pointwise coadjoint evolution. The exponential tilt therefore couples
to the corresponding central component of the momentum. More general
fields $\lambda(\cdot)$ represent spatially dependent currents or
chemical potentials, but they need not define conserved observables
without additional symmetry assumptions.
\end{rem}

\subsection{$H^s$ right-invariant metrics on $\operatorname{Diff}(M)$
and the EPDiff equation}

\subsubsection*{Geometry and geodesic equation}

Let $\operatorname{Diff}(M)$ be the Fréchet--Lie group of
diffeomorphisms of $M$, and let $\mathfrak X(M)$ be its Lie algebra of
smooth vector fields. Let
\[
A_s\colon\mathfrak X(M)\longrightarrow\mathfrak X(M)
\]
be a positive selfadjoint elliptic operator with respect to the
$L^2$-pairing induced by $g$, for instance
\[
A_s=(1+\Delta_g)^s.
\]
The inner product at the identity is
\[
\langle u,v\rangle_{H^s}
=
\int_M
\langle A_s^{1/2}u,A_s^{1/2}v\rangle_g
\,\operatorname{vol}_g,
\]
and its right translations define a right-invariant weak Riemannian
metric on $\operatorname{Diff}(M)$.

The momentum is naturally a one-form density,
\[
m=(A_su)^\flat\otimes\operatorname{vol}_g.
\]
The Euler--Poincaré equation is
\[
\partial_t m+\mathcal L_u m=0,
\]
where $\mathcal L_u$ denotes the Lie derivative of one-form densities.
If the density is written relative to the fixed volume form and $m$ is
identified with its one-form component, the same equation becomes
\[
\partial_t m+\mathcal L_u m+(\operatorname{div}_g u)m=0,
\qquad
m=(A_su)^\flat.
\]

In local coordinates, for
\[
A_s=\operatorname{id}-\alpha^2\Delta_{\mathrm{geom}},
\]
where $\Delta_{\mathrm{geom}}$ denotes the nonpositive coordinate
Laplacian, the equation takes the familiar form
\[
\partial_t m
+
u\cdot\nabla m
+
(\nabla u)^\top m
+
m(\nabla\cdot u)
=
0,
\qquad
m=u-\alpha^2\Delta_{\mathrm{geom}}u.
\]
In one space dimension this includes the Camassa--Holm equation, while
in higher dimension it gives the EPDiff system.

\subsubsection*{Hamiltonian and symmetry momenta}

The reduced Hamiltonian is
\[
H(u)
=
\frac12
\int_M
\langle u,A_su\rangle_g
\,\operatorname{vol}_g.
\]
Let
\[
K\subset\operatorname{Diff}(M)
\]
be a compact group of isometries of $(M,g)$, with Lie algebra
$\mathfrak k$, and assume that $A_s$ is $K$-equivariant. For each
Killing field $w\in\mathfrak k$, the corresponding momentum observable
is
\[
\Phi_w(u)
=
\int_M
\langle A_su,w\rangle_g
\,\operatorname{vol}_g
=
\int_M
m(w).
\]
Under the stated equivariance hypothesis, these observables arise from
the cotangent-lifted symmetry and are conserved by the geodesic flow.
They are therefore natural constraints in the Souriau-type variational
principle.

\subsubsection*{Equilibrium means}

Let $\mathcal A_b$ consist of bounded continuous cylindrical
functionals of $u$, and let $\mathcal L$ also contain $H$ and the
momentum observables $\Phi_w$. Let $\mathsf m_0$ be a Gaussian reference
mean with covariance related to $A_s^{-1}$ when this defines a Gaussian
Radon measure on the chosen completion, or a normalized cut-off mean.

Choose a basis
\[
\mathcal B=(w_1,\ldots,w_N)
\]
of $\mathfrak k$ and write
\[
\lambda=(\lambda^1,\ldots,\lambda^N)\in\mathbb R^N.
\]
The variational functional is
\[
\mathcal G_{\beta,\lambda}(\mathsf n)
=
\mathcal H(\mathsf n\,\|\,\mathsf m_0)
+
\beta\,\mathsf n(H)
+
\sum_{j=1}^{N}
\lambda^j\mathsf n(\Phi_{w_j}).
\]
Whenever the potential
\[
F_{\beta,\lambda}
=
-\beta H-\sum_{j=1}^{N}\lambda^j\Phi_{w_j}
\]
belongs to $\mathcal E(\mathsf m_0)$ and satisfies the
translation-compatibility and separation hypotheses of
Theorem~\ref{thm:existence}, the unique equilibrium mean is
\[
\boxed{
\mathsf n_{\beta,\lambda}^*(f)
=
\frac{
\mathsf m_0
\left(
f e^{-\beta H-\sum_{j=1}^{N}\lambda^j\Phi_{w_j}}
\right)
}{
\mathsf m_0
\left(
e^{-\beta H-\sum_{j=1}^{N}\lambda^j\Phi_{w_j}}
\right)
},
\qquad
f\in\mathcal L.
}
\]
If the $H^s$ geodesic flow is globally defined and preserves the
reference mean and the observable algebra, then the conservation of the
symmetry momenta implies that $\mathsf n_{\beta,\lambda}^*$ is
stationary. Under the additional assumptions of
Theorem~\ref{thm:poisson-kms}, it also satisfies the classical
Poisson--KMS identity.

\begin{rem}[Vorticity and Casimir constraints]
In dimension two, the incompressible Euler equation arises from the
$L^2$ metric on the volume-preserving diffeomorphism group
$\operatorname{Diff}_{\operatorname{vol}}(M)$, rather than on the full
group $\operatorname{Diff}(M)$. In this setting, one may enlarge the
constraint family by adding suitable Casimir functionals of the
vorticity. At the formal mean-field level, the resulting entropy
principle leads to relations of Miller--Robert--Sommeria type. A
complete identification with MRS theory requires a separate treatment
of local vorticity-level distributions and of the relevant
large-deviation principle.
\end{rem}

\subsection{Compatibility with the Fréchet--Hilbert scheme}

The preceding geometric constructions are formulated on smooth
Fréchet manifolds. For analytic purposes, one may pass to Sobolev
completions
\[
H^r(M,G),
\qquad
\operatorname{Diff}^{\,r}(M),
\qquad
r>\frac d2+1,
\]
using an index $r$ sufficiently large for the relevant manifold and
composition structures. The order $s$ of the inertia operator and the
regularity index $r$ of the Hilbert completion need not coincide.

On such Hilbert manifolds, local well-posedness of the geodesic equation
follows when the corresponding geodesic spray extends smoothly to the
chosen completion. This property depends on the order of the inertia
operator and is not automatic from ellipticity alone. A Gaussian
reference may be used only when its covariance defines a Radon Gaussian
measure on the chosen model space. Otherwise, one must pass to a weaker
distribution space or use finite-dimensional spectral cut-offs.

Returning to the smooth Fréchet level requires compatibility of the
solutions, observable algebras, and normalized means across the
Sobolev scale. It is therefore an additional projective-limit
hypothesis rather than an automatic consequence of Hilbert-level
well-posedness.
\section{Concrete Examples: Current Groups and Diffeomorphism Groups}

This section spells out the preceding abstract construction in concrete
Hamiltonian systems. Its purpose is not to develop a complete analytic
theory for each equation, but to identify the Hamiltonian, the
constraint observables, possible reference means, and the resulting
conditional equilibrium construction.

The Hamiltonian descriptions below are understood either on the full
cotangent phase space or, after reduction, on a suitable symplectic
leaf of the corresponding Lie--Poisson space. Throughout,
$\mathcal A_b$ denotes the algebra of bounded cylindrical test
observables, while $\mathcal L$ denotes the larger algebra containing
the Hamiltonian, the constraints, and the admissible exponential
weights.
\subsection{Example A: the Euler top as a finite-dimensional benchmark}
\label{subsec:euler-top-example}

The classical Euler top provides a finite-dimensional benchmark for the
normalized-mean formalism. It may also be regarded as the special case
of the current-group construction in which the base manifold consists
of a single point.

Let
\[
G=\operatorname{SO}(3),
\qquad
\mathfrak{so}(3)^*\simeq\mathbb R^3,
\]
and let
\[
I\colon\mathbb R^3\longrightarrow\mathbb R^3
\]
be a positive-definite inertia operator. If
$m\in\mathbb R^3$ denotes the body angular momentum and
\[
\Omega=I^{-1}m
\]
the angular velocity, then the reduced Hamiltonian is
\[
H(m)
=
\frac12\,m\cdot I^{-1}m.
\]
With the usual identification of the Lie bracket on
$\mathfrak{so}(3)$ with the vector product, the Euler equation is
\[
\dot m=m\times\Omega.
\]

The Casimir
\[
C(m)=|m|^2
\]
is conserved. For a fixed value $R>0$, the corresponding coadjoint
orbit is the sphere
\[
\mathcal O_R
=
\left\{
m\in\mathbb R^3
\;\middle|\;
|m|=R
\right\}.
\]
It is a compact symplectic manifold equipped with its
Kostant--Kirillov--Souriau form. Let $\mu_R$ denote the associated
normalized Liouville measure, and define the reference mean
\[
\mathsf m_0(f)
=
\int_{\mathcal O_R}f(m)\,\mathrm d\mu_R(m),
\qquad
f\in C(\mathcal O_R).
\]
This mean is faithful, invariant under the coadjoint action, and
Poisson invariant.

The moment map for the coadjoint action is the inclusion
\[
\mathbf J\colon\mathcal O_R\longrightarrow\mathfrak{so}(3)^*,
\qquad
\mathbf J(m)=m.
\]
For
\[
\lambda\in\mathfrak{so}(3)\simeq\mathbb R^3,
\]
consider the potential
\[
F_{\beta,\lambda}(m)
=
-\beta H(m)-\lambda\cdot m.
\]
Since $\mathcal O_R$ is compact, every continuous potential is
exponentially admissible and
\[
0<
Z(\beta,\lambda)
=
\int_{\mathcal O_R}
e^{-\beta H(m)-\lambda\cdot m}
\,\mathrm d\mu_R(m)
<\infty.
\]
Moreover, translation compatibility and the separation hypotheses are
automatic on the algebra $C(\mathcal O_R)$. Theorem~\ref{thm:existence}
therefore gives the unique equilibrium state
\[
\boxed{
\mathsf n_{\beta,\lambda}^*(f)
=
\frac{
\displaystyle
\int_{\mathcal O_R}
f(m)e^{-\beta H(m)-\lambda\cdot m}
\,\mathrm d\mu_R(m)
}{
\displaystyle
\int_{\mathcal O_R}
e^{-\beta H(m)-\lambda\cdot m}
\,\mathrm d\mu_R(m)
},
\qquad
f\in C(\mathcal O_R).
}
\]
Thus, in this finite-dimensional case, the equilibrium mean is
represented by the classical Gibbs probability measure
\[
\mathrm d\nu_{\beta,\lambda}^*(m)
=
\frac{1}{Z(\beta,\lambda)}
e^{-\beta H(m)-\lambda\cdot m}
\,\mathrm d\mu_R(m).
\]

The logarithmic partition functional
\[
\psi(\beta,\lambda)
=
\log Z(\beta,\lambda)
\]
is smooth and convex. Its first derivatives are
\[
\partial_\beta\psi(\beta,\lambda)
=
-\mathsf n_{\beta,\lambda}^*(H)
\]
and
\[
D_\lambda\psi(\beta,\lambda)[\dot\lambda]
=
-\mathsf n_{\beta,\lambda}^*(\dot\lambda\cdot m).
\]
Its Hessian is the covariance form of the observables
\[
H,\quad m_1,\quad m_2,\quad m_3.
\]
It is positive definite after quotienting by every linear combination
of these observables that is constant on $\mathcal O_R$.

For $\lambda=0$, the Gibbs weight depends only on the conserved
Hamiltonian. Hence
\[
\mathsf n_{\beta,0}^*
\]
is stationary under the Euler flow. Since $\mu_R$ is Poisson invariant,
it also satisfies
\[
\mathsf n_{\beta,0}^*(\{f,g\})
=
\beta\,
\mathsf n_{\beta,0}^*
\bigl(g\{f,H\}\bigr)
\]
for all smooth observables $f,g$ on $\mathcal O_R$.

For $\lambda\neq0$, the generalized equilibrium satisfies instead
\[
\mathsf n_{\beta,\lambda}^*(\{f,g\})
=
\beta\,
\mathsf n_{\beta,\lambda}^*
\bigl(g\{f,H\}\bigr)
+
\mathsf n_{\beta,\lambda}^*
\bigl(g\{f,\lambda\cdot m\}\bigr).
\]
The corresponding state is stationary under the Euler flow only when
the additional observable
\[
m\longmapsto\lambda\cdot m
\]
is conserved. This occurs when $\lambda$ generates a symmetry of the
inertia operator; for a generic asymmetric top, no nonzero linear
momentum component is conserved in body coordinates.

This example shows that the abstract normalized-mean construction
reduces, on a compact coadjoint orbit, to Souriau's
finite-dimensional Gibbs formalism. It also illustrates the distinction
between a generalized exponential family and a stationary equilibrium
for a prescribed Hamiltonian evolution.
\subsection{Example B: the current group on the circle}

Let
\[
M=\mathbb S^1,
\qquad
G=\operatorname{SU}(2),
\]
and consider the Fréchet--Lie group
\[
\mathcal G
=
\operatorname{Map}(\mathbb S^1,\operatorname{SU}(2)).
\]
Its Lie algebra is
\[
\mathfrak X_G
=
C^\infty(\mathbb S^1,\mathfrak{su}(2)),
\]
with pointwise bracket.

\paragraph{\bfseries Geometry and geodesic equation.}
Fix $s>0$ and define
\[
A_s
=
(1-\partial_x^2)^s
\otimes
\operatorname{id}_{\mathfrak{su}(2)}.
\]
The right-invariant $H^s$ metric is determined at the identity by
\[
\langle\xi,\eta\rangle_{H^s}
=
\int_{\mathbb S^1}
\left\langle
A_s^{1/2}\xi(x),A_s^{1/2}\eta(x)
\right\rangle_{\mathfrak{su}(2)}
\,\mathrm dx.
\]
Let $\xi(t,x)\in\mathfrak{su}(2)$ be the right-trivialized velocity and
set
\[
m=A_s\xi.
\]
With the convention
\[
\langle\operatorname{ad}_\xi^*m,\eta\rangle
=
\langle m,[\eta,\xi]\rangle,
\]
the Euler--Poincaré equation is
\[
\partial_t m+\operatorname{ad}_\xi^*m=0.
\]
After identifying
\[
\mathfrak{su}(2)^*\simeq\mathfrak{su}(2)
\]
by an $\operatorname{Ad}$-invariant inner product, this becomes
\[
\partial_t m+[m,\xi]=0
\]
with the preceding convention.

\paragraph{\bfseries Reference mean and observables.}
Let $\mathcal A_b$ be generated by bounded cylindrical functionals of
the Fourier modes of $\xi$, and let $\mathcal L$ additionally contain
the Sobolev energy
\[
H(\xi)
=
\frac12
\int_{\mathbb S^1}
\langle\xi,A_s\xi\rangle_{\mathfrak{su}(2)}
\,\mathrm dx.
\]
As reference mean $\mathsf m_0$, one may use a centered Gaussian
reference with covariance $A_s^{-1}$ when it defines a Gaussian Radon
measure on the chosen completion, or a normalized cut-off mean obtained
from spectral truncations.

On the $L^2$-based Hilbert space,
$A_s^{-1}$ is trace class in dimension one when
\[
s>\frac12.
\]
For smaller $s$, the Gaussian reference must be realized on a weaker
distribution space or understood through finite-dimensional cut-offs.

\paragraph{\bfseries Constraints and equilibrium.}
For
\[
\lambda\in
C^\infty(\mathbb S^1,\mathfrak{su}(2)),
\]
define the linear momentum observable
\[
\langle\lambda,m\rangle
=
\int_{\mathbb S^1}
\langle\lambda(x),m(x)\rangle_{\mathfrak{su}(2)}
\,\mathrm dx.
\]
Whenever the potential
\[
F_{\beta,\lambda}
=
-\beta H-\langle\lambda,m\rangle
\]
belongs to $\mathcal E(\mathsf m_0)$ and satisfies the
translation-compatibility and separation hypotheses of
Theorem~\ref{thm:existence}, the equilibrium mean is
\[
\boxed{
\mathsf n_{\beta,\lambda}^*(f)
=
\frac{
\mathsf m_0
\left(
f e^{-\beta H-\langle\lambda,m\rangle}
\right)
}{
\mathsf m_0
\left(
e^{-\beta H-\langle\lambda,m\rangle}
\right)
},
\qquad
f\in\mathcal L.
}
\]
A constant field $\lambda$ couples the equilibrium weight to the
corresponding spatially integrated momentum component. A nonconstant
field $\lambda(x)$ introduces a spatially dependent current or chemical
potential. Such a linear observable is conserved only when it is
generated by a symmetry of the dynamics.

\begin{rem}
Since $\mathfrak{su}(2)$ is simple, its centre is zero. Hence there is
no nonzero central direction for which the associated linear momentum
observable would be conserved solely by centrality.
\end{rem}

\paragraph{\bfseries Partition functional.}
Writing formally
\[
\xi(x)
=
\sum_{k\in\mathbb Z}
\widehat\xi_k e^{ikx},
\]
with the appropriate reality condition, the Gaussian covariance is
\[
\mathbb E
\left[
\widehat\xi_k\otimes\widehat\xi_{-k}
\right]
=
(1+k^2)^{-s}\operatorname{id}_{\mathfrak{su}(2)}.
\]
At finite spectral cut-off, the partition functional may be computed
mode by mode. Its infinite-dimensional limit, or any required
renormalized interpretation, must be checked separately. Whenever the
local exponential regularity hypothesis holds, the Hessian of the
logarithmic partition functional is the covariance form. It is positive
definite on every finite-dimensional constraint subspace on which that
covariance form is nondegenerate.

\subsection{Example C: Camassa--Holm and EPDiff in one space dimension}

\paragraph{\bfseries Geometry.}
Let $M=\mathbb S^1$ and consider
$\operatorname{Diff}(\mathbb S^1)$. For the $H^1$ metric, take
\[
A=1-\alpha^2\partial_x^2,
\qquad
m=Au.
\]
The Euler--Poincaré equation becomes the Camassa--Holm equation in
momentum form:
\[
\partial_t m+u\partial_xm+2u_xm=0,
\qquad
m=u-\alpha^2u_{xx}.
\]
The Hamiltonian is
\[
H(u)
=
\frac12
\int_{\mathbb S^1}um\,\mathrm dx
=
\frac12
\int_{\mathbb S^1}
\left(
u^2+\alpha^2u_x^2
\right)
\,\mathrm dx.
\]

\paragraph{\bfseries Symmetry, constraint and equilibrium.}
The rotation group
\[
K\simeq\mathbb S^1
\]
acts by translations and preserves the inertia operator $A$. Its
infinitesimal generator is the constant Killing field, customarily
represented by $w=\partial_x$. The corresponding conserved momentum
observable is
\[
\Phi_w(u)
=
\int_{\mathbb S^1}m\,\mathrm dx
=
\int_{\mathbb S^1}u\,\mathrm dx.
\]

Let $\mathcal A_b$ consist of bounded cylindrical functionals of the
Fourier modes of $u$, and let $\mathcal L$ also contain $H$ and
$\Phi_w$. Let $\mathsf m_0$ be a Gaussian reference with covariance
$A^{-1}$ on a suitable Hilbert or distribution space, or a normalized
spectral cut-off mean. Whenever
\[
F_{\beta,\lambda}
=
-\beta H-\lambda\Phi_w
\]
belongs to $\mathcal E(\mathsf m_0)$ and satisfies the
translation-compatibility and separation hypotheses of
Theorem~\ref{thm:existence}, the equilibrium mean is
\[
\boxed{
\mathsf n_{\beta,\lambda}^*(f)
=
\frac{
\mathsf m_0
\left(
f e^{-\beta H-\lambda\Phi_w}
\right)
}{
\mathsf m_0
\left(
e^{-\beta H-\lambda\Phi_w}
\right)
},
\qquad
f\in\mathcal L.
}
\]
If the $H^1$ geodesic flow is global on the domain under consideration
and preserves $\mathsf m_0$ and $\mathcal L$, then conservation of
$\Phi_w$ implies that $\mathsf n_{\beta,\lambda}^*$ is stationary. If
the extended Poisson-invariance assumptions of
Theorem~\ref{thm:poisson-kms} also hold, then the classical
Poisson--KMS identity follows.

\paragraph{\bfseries Peakon sector.}
For $\alpha>0$, the Camassa--Holm equation admits peakon solutions of
the form
\[
u(t,x)
=
\sum_{j=1}^{N}
p_j(t)K_\alpha(x-q_j(t)),
\]
where $K_\alpha$ is the periodic Green kernel of $A$. On the
finite-dimensional sector of distinct peakon positions, the Hamiltonian
becomes
\[
H_N(q,p)
=
\frac12
\sum_{i,j=1}^{N}
p_ip_jK_\alpha(q_i-q_j).
\]
Whenever the corresponding finite-dimensional partition functional is
finite and nonzero, exponential reweighting defines a Gibbs-type
equilibrium on the peakon variables. Compatibility with the
infinite-dimensional construction requires an additional convergence
argument and is not automatic.

\subsection{Example D: EPDiff on the flat torus}

\paragraph{\bfseries Geometry and equation.}
Let $M=\mathbb T^d$ and let
\[
A_s=(1-\Delta_{\mathrm{geom}})^s
\]
act componentwise on vector fields
\[
u\colon\mathbb T^d\longrightarrow\mathbb R^d,
\]
where $\Delta_{\mathrm{geom}}$ is the nonpositive flat Laplacian. With
\[
m=A_su,
\]
the EPDiff equation reads
\[
\partial_t m
+
u\cdot\nabla m
+
(\nabla u)^\top m
+
m(\nabla\cdot u)
=
0.
\]
Here $m$ denotes the one-form component of the corresponding one-form
density. The Hamiltonian is
\[
H(u)
=
\frac12
\int_{\mathbb T^d}
u\cdot A_su\,\mathrm dx.
\]

\paragraph{\bfseries Translation constraints.}
The torus translation group $K=\mathbb T^d$ preserves $A_s$ and yields
the conserved linear momenta
\[
\Phi_{e_j}(u)
=
\int_{\mathbb T^d}m_j(x)\,\mathrm dx
=
\int_{\mathbb T^d}u_j(x)\,\mathrm dx,
\qquad
j=1,\ldots,d.
\]
The last equality follows because the zero Fourier mode of $A_s$ has
eigenvalue one.

Let $\mathcal A_b$ consist of bounded cylindrical functionals of $u$,
and let $\mathcal L$ contain $H$ and the translation constraints. For
\[
\lambda=(\lambda_1,\ldots,\lambda_d)\in\mathbb R^d,
\]
set
\[
F_{\beta,\lambda}
=
-\beta H-\sum_{j=1}^{d}\lambda_j\Phi_{e_j}.
\]
Whenever $F_{\beta,\lambda}$ belongs to
$\mathcal E(\mathsf m_0)$ and satisfies the
translation-compatibility and separation hypotheses of
Theorem~\ref{thm:existence}, the equilibrium mean is
\[
\boxed{
\mathsf n_{\beta,\lambda}^*(f)
=
\frac{
\mathsf m_0
\left(
f e^{-\beta H-\sum_{j=1}^{d}\lambda_j\Phi_{e_j}}
\right)
}{
\mathsf m_0
\left(
e^{-\beta H-\sum_{j=1}^{d}\lambda_j\Phi_{e_j}}
\right)
},
\qquad
f\in\mathcal L.
}
\]

At finite Gaussian cut-off, the partition functional factorizes over
independent Fourier modes because the quadratic Hamiltonian is diagonal
in Fourier variables. In the infinite-dimensional limit, convergence
or renormalization must be verified. Under local exponential
regularity, strict convexity on a finite-dimensional constraint
subspace follows when the corresponding covariance matrix is positive
definite.

\subsection{Example E: two-dimensional Euler and vorticity constraints}

\paragraph{\bfseries Geometry and vorticity form.}
Let $M=\mathbb T^2$, restrict to divergence-free vector fields of zero
mean, and use the $L^2$ metric on the volume-preserving diffeomorphism
group. The Euler equation is
\[
\partial_t\omega+u\cdot\nabla\omega=0,
\qquad
u=\nabla^\perp\psi,
\qquad
-\Delta\psi=\omega.
\]
The kinetic energy is
\[
H(u)
=
\frac12
\int_{\mathbb T^2}|u|^2\,\mathrm dx
=
\frac12
\int_{\mathbb T^2}\omega\psi\,\mathrm dx.
\]

\paragraph{\bfseries Casimirs and equilibrium relations.}
For every sufficiently regular function $\phi$, the functional
\[
C_\phi(\omega)
=
\int_{\mathbb T^2}\phi(\omega)\,\mathrm dx
\]
is formally conserved by a sufficiently regular Euler flow. Fix
\[
\phi_1,\ldots,\phi_K
\]
and consider the constraints $C_{\phi_k}$. Let $\mathcal A_b$ consist
of bounded cylindrical vorticity observables, and let $\mathcal L$
contain the energy, the chosen Casimirs, and the required exponential
weights.

Let $\mathsf m_0$ be a Gaussian reference mean on a suitable
distribution space or a normalized spectral cut-off mean. The
variational functional is
\[
\mathcal G_{\beta,\alpha}(\mathsf n)
=
\mathcal H(\mathsf n\,\|\,\mathsf m_0)
+
\beta\,\mathsf n(H)
+
\sum_{k=1}^{K}
\alpha_k\mathsf n(C_{\phi_k}).
\]
Whenever
\[
F_{\beta,\alpha}
=
-\beta H-\sum_{k=1}^{K}\alpha_kC_{\phi_k}
\]
belongs to $\mathcal E(\mathsf m_0)$ and satisfies the
translation-compatibility and separation hypotheses of
Theorem~\ref{thm:existence}, the equilibrium mean is
\[
\boxed{
\mathsf n_{\beta,\alpha}^*(f)
=
\frac{
\mathsf m_0
\left(
f e^{-\beta H-\sum_{k=1}^{K}\alpha_kC_{\phi_k}}
\right)
}{
\mathsf m_0
\left(
e^{-\beta H-\sum_{k=1}^{K}\alpha_kC_{\phi_k}}
\right)
},
\qquad
f\in\mathcal L.
}
\]

If $\mathsf m_0$ is represented by a probability measure $\mu_0$, the
same formula may be written
\[
\mathrm d\nu_{\beta,\alpha}^*(\omega)
=
\frac{1}{Z(\beta,\alpha)}
\exp\left(
-\beta H(\omega)
-\sum_{k=1}^{K}\alpha_kC_{\phi_k}(\omega)
\right)
\,\mathrm d\mu_0(\omega).
\]

A formal critical-point calculation with finitely many differentiable
Casimir constraints gives a relation of the form
\[
\beta\psi
+
\sum_{k=1}^{K}\alpha_k\phi_k'(\omega)
=
0.
\]
When the resulting relation can be inverted, it yields
\[
\omega=F(\psi).
\]
In a mean-field interpretation, an analogous relation may hold for the
coarse-grained vorticity $\overline\omega$. Such relations are
reminiscent of Miller--Robert--Sommeria equilibria, but a complete
derivation of MRS theory additionally requires local probability
distributions of vorticity levels, mixing entropy, and the associated
large-deviation or variational analysis.

\subsection{Example F: semidirect current-coupled EPDiff}

\paragraph{\bfseries Set-up.}
Let $M$ be compact and let $\mathcal K$ be a compact Lie group with Lie
algebra $\mathfrak k$. Consider the semidirect product
\[
\operatorname{Diff}(M)
\ltimes
\operatorname{Map}(M,\mathcal K),
\]
where the action defining the semidirect product is understood to be
the natural action of diffeomorphisms on current fields.

Its Euler--Poincaré equations couple an EPDiff-type velocity $u$ with a
current variable
\[
\zeta(x)\in\mathfrak k.
\]
A typical quadratic Hamiltonian has the form
\[
H(u,\zeta)
=
\frac12
\int_M
\left(
\langle u,A_s^{u}u\rangle
+
\langle\zeta,A_s^{\zeta}\zeta\rangle
\right)
\,\operatorname{vol}_M,
\]
where $A_s^{u}$ and $A_s^{\zeta}$ are suitable positive elliptic inertia
operators on vector fields and current fields, respectively. The
associated momentum variables are
\[
m=A_s^{u}u,
\qquad
\rho=A_s^{\zeta}\zeta.
\]
They should not, in general, be identified with a single momentum map:
the precise momentum maps and conserved quantities depend on the
chosen action and its symmetries.

Let $\Phi_1,\ldots,\Phi_N$ be conserved observables arising from a
finite-dimensional symmetry subgroup of the coupled system. Let
$\mathcal A_b$ consist of bounded cylindrical functionals of
$(u,\zeta)$, and let $\mathcal L$ additionally contain $H$, the
observables $\Phi_a$, and the required exponential weights. Set
\[
F_{\beta,\lambda}
=
-\beta H-\sum_{a=1}^{N}\lambda^a\Phi_a.
\]
Whenever $F_{\beta,\lambda}$ belongs to
$\mathcal E(\mathsf m_0)$ and satisfies the
translation-compatibility and separation hypotheses of
Theorem~\ref{thm:existence}, the entropy principle gives
\[
\boxed{
\mathsf n_{\beta,\lambda}^*(f)
=
\frac{
\mathsf m_0
\left(
f e^{-\beta H-\sum_{a=1}^{N}\lambda^a\Phi_a}
\right)
}{
\mathsf m_0
\left(
e^{-\beta H-\sum_{a=1}^{N}\lambda^a\Phi_a}
\right)
},
\qquad
f\in\mathcal L.
}
\]
This shows that the conditional equilibrium construction extends to
semidirect products, provided the analytic and dynamical hypotheses of
the abstract theory are verified for the coupled system.

\subsection{Summary of the examples}

The examples illustrate complementary aspects of the normalized-mean
Souriau framework. For current groups, the reduced momentum is a field
of Lie-algebra-valued currents, and linear observables are obtained by
pairing these currents with test fields. Such observables define
stationary constraints only when they arise from symmetries or are
otherwise conserved. For diffeomorphism groups, EPDiff and
Camassa--Holm provide geometric PDEs whose linear momenta arise from
compact symmetry groups. For two-dimensional Euler, energy and Casimir
constraints lead formally to equilibrium relations related to those of
Miller--Robert--Sommeria theory.

In every case, the candidate equilibrium is obtained by exponential
reweighting of a reference mean by the Hamiltonian and the relevant
symmetry or Casimir constraints. The abstract theory applies once
exponential admissibility, translation compatibility, separation,
local exponential regularity, and the required dynamical invariance
properties have been verified in the model under consideration.

\section{Full proofs of the main statements}
\label{sec:full-proofs}

We collect here the proofs of the statements used in the previous
sections. Throughout this section, $\mathcal L$ denotes the unital vector
lattice algebra introduced in Section~1, and $\mathsf m_0$ is a normalized
mean on $\mathcal L$. The Hamiltonian $H$ and the finite family of
constraint observables
\[
\Phi_1,\ldots,\Phi_N
\]
belong to $\mathcal L$ and may be unbounded.

For
\[
(\beta,\lambda)\in(0,\infty)\times\mathbb R^N,
\]
we write
\[
F_{\beta,\lambda}
=
-\beta H-\sum_{a=1}^{N}\lambda^a\Phi_a.
\]

We recall the standing hypotheses:
\begin{enumerate}[label=(A\arabic*)]
\item \textbf{Exponential admissibility and local regularity.}
For every admissible parameter $(\beta,\lambda)\in\mathcal D$,
\[
F_{\beta,\lambda}\in\mathcal E(\mathsf m_0),
\]
and the local exponential regularity condition of
Definition~\ref{def:local-exponential-regularity} holds.

\item \textbf{Coercivity.}
There exists a proper functional
\[
V\colon\mathcal M\longrightarrow[0,\infty)
\]
and constants $c_1,c_2>0$ such that
\[
H(m)\geq c_1V(m)-c_2,
\qquad
m\in\mathcal M,
\]
and the sublevel sets of $\mathcal G_{\beta,\lambda}$ are compact for
the topology of pointwise convergence on $\mathcal L$.

\item \textbf{Covariance non-degeneracy.}
For every admissible parameter $\theta=(\beta,\lambda)$ and every
parameter direction
\[
\dot\theta=(\dot\beta,\dot\lambda),
\]
the covariance form satisfies
\[
\operatorname{Var}_{\mathsf n_\theta^*}
\left(
\dot\beta H+\sum_{a=1}^{N}\dot\lambda^a\Phi_a
\right)>0
\]
unless $\dot\theta$ is thermodynamically null at $\theta$.

\item \textbf{Compatibility under exponential tilting.}
For every admissible thermodynamic parameter $(\beta,\lambda)$, the
potential
\[
F_{\beta,\lambda}
=
-\beta H-\sum_{a=1}^{N}\lambda^a\Phi_a
\]
is translation compatible with $\mathsf m_0$ in the sense of
Definition~\ref{def:translation-compatibility}.
\end{enumerate}

Recall that
\[
\mathsf S(\mathcal L)
=
\left\{
\mathsf n\colon\mathcal L\to\mathbb R
\;\middle|\;
\mathsf n\text{ is linear},
\quad
\mathsf n(f)\geq 0\ \text{for }f\geq 0,
\quad
\mathsf n(1)=1
\right\}.
\]

\subsection{Proof of Proposition~\ref{prop:entropy-properties}}

Let $\mathsf m_0\in\mathsf S(\mathcal L)$ and define
\[
\Lambda_{\mathsf m_0}(f)
=
\log\mathsf m_0(e^f),
\qquad
f\in\mathcal E(\mathsf m_0).
\]
The relative entropy is
\[
\mathcal H(\mathsf n\,\|\,\mathsf m_0)
=
\sup_{f\in\mathcal E(\mathsf m_0)}
\left\{
\mathsf n(f)-\Lambda_{\mathsf m_0}(f)
\right\}.
\]

\medskip
\noindent
\textit{Non-negativity.}
Taking $f=0$ gives
\[
\mathcal H(\mathsf n\,\|\,\mathsf m_0)
\geq
\mathsf n(0)-\log\mathsf m_0(1)
=
0.
\]

\medskip
\noindent
\textit{Separation.}
Assume that
\[
\mathcal H(\mathsf n\,\|\,\mathsf m_0)=0.
\]
Then, for every admissible $f$,
\[
\mathsf n(f)\leq\log\mathsf m_0(e^f).
\]
Let $h\in\mathcal E(\mathsf m_0)$ be such that $th$ is admissible for
all sufficiently small $|t|$. Applying the preceding inequality to
$th$, with $t>0$, gives
\[
\mathsf n(h)
\leq
\frac{1}{t}\log\mathsf m_0(e^{th}).
\]
Local exponential regularity at $t=0$ gives
\[
\left.
\frac{\mathrm d}{\mathrm dt}
\log\mathsf m_0(e^{th})
\right|_{t=0}
=
\mathsf m_0(h).
\]
Consequently, letting $t\downarrow0$,
\[
\mathsf n(h)\leq\mathsf m_0(h).
\]
Applying the same argument to $-h$ gives
\[
\mathsf n(h)\geq\mathsf m_0(h).
\]
Therefore
\[
\mathsf n(h)=\mathsf m_0(h)
\]
for every admissible $h$. Since the admissible potentials separate
normalized means on $\mathcal L$, we obtain
\[
\mathsf n=\mathsf m_0.
\]

Conversely, if $\mathsf n=\mathsf m_0$, Jensen's inequality for positive
normalized linear functionals gives
\[
\mathsf m_0(f)
\leq
\log\mathsf m_0(e^f).
\]
Therefore
\[
\mathsf m_0(f)-\log\mathsf m_0(e^f)\leq 0
\]
for every admissible $f$. Taking the supremum and using the test
function $f=0$ gives
\[
\mathcal H(\mathsf m_0\,\|\,\mathsf m_0)=0.
\]

\medskip
\noindent
\textit{Convexity and lower semicontinuity.}
For fixed $\mathsf m_0$, the map
\[
\mathsf n
\longmapsto
\mathsf n(f)-\log\mathsf m_0(e^f)
\]
is affine and continuous for the topology of pointwise convergence on
$\mathcal L$. Since $\mathcal H(\mathsf n\,\|\,\mathsf m_0)$ is the
supremum of these affine continuous functions, it is convex and lower
semicontinuous for that topology.

This proves Proposition~\ref{prop:entropy-properties}.

\subsection{Change of reference under exponential tilting}
\label{subsec:change-reference-proof}

The proof of uniqueness requires the following change-of-reference
identity.

\begin{Proposition}[Change-of-reference identity]
\label{prop:change-reference}
Let $F\in\mathcal E(\mathsf m_0)$ and define
\[
\mathsf m_F(f)
=
\frac{\mathsf m_0(fe^F)}{\mathsf m_0(e^F)}.
\]
Assume that the admissible classes are compatible under exponential
tilting:
\[
g\in\mathcal E(\mathsf m_F)
\quad\Longleftrightarrow\quad
g+F\in\mathcal E(\mathsf m_0).
\]
Then, for every $\mathsf n\in\mathsf S(\mathcal L)$ for which the
quantities below are defined,
\[
\mathcal H(\mathsf n\,\|\,\mathsf m_F)
=
\mathcal H(\mathsf n\,\|\,\mathsf m_0)
-
\mathsf n(F)
+
\log\mathsf m_0(e^F).
\]
\end{Proposition}

\begin{proof}
For every $g\in\mathcal E(\mathsf m_F)$,
\[
\mathsf m_F(e^g)
=
\frac{\mathsf m_0(e^{g+F})}{\mathsf m_0(e^F)}.
\]
Therefore
\begin{align*}
\mathcal H(\mathsf n\,\|\,\mathsf m_F)
&=
\sup_{g\in\mathcal E(\mathsf m_F)}
\left\{
\mathsf n(g)-\log\mathsf m_F(e^g)
\right\}
\\
&=
\sup_{g\in\mathcal E(\mathsf m_F)}
\left\{
\mathsf n(g)
-
\log\mathsf m_0(e^{g+F})
+
\log\mathsf m_0(e^F)
\right\}.
\end{align*}
Set
\[
h=g+F.
\]
By the compatibility assumption, this change of variables identifies
$\mathcal E(\mathsf m_F)+F$ with $\mathcal E(\mathsf m_0)$. Hence
\begin{align*}
\mathcal H(\mathsf n\,\|\,\mathsf m_F)
&=
\sup_{h\in\mathcal E(\mathsf m_0)}
\left\{
\mathsf n(h-F)
-
\log\mathsf m_0(e^h)
+
\log\mathsf m_0(e^F)
\right\}
\\
&=
\mathcal H(\mathsf n\,\|\,\mathsf m_0)
-
\mathsf n(F)
+
\log\mathsf m_0(e^F).
\end{align*}
\end{proof}

\subsection{Proof of Theorem~\ref{thm:partition-positive}}

By definition,
\[
Z(\beta,\lambda)
=
\mathsf m_0(e^{F_{\beta,\lambda}}).
\]
Finiteness follows from exponential admissibility.

Strict positivity is not a formal consequence of positivity alone in
the cut-off situation: a limit of strictly positive normalized
quotients may vanish. It must therefore follow either from faithfulness
of $\mathsf m_0$ or from a uniform lower bound for the truncated
partition functionals.

If $\mathsf m_0$ is faithful, then
\[
e^{F_{\beta,\lambda}}>0
\]
pointwise and is not identically zero. Hence
\[
\mathsf m_0(e^{F_{\beta,\lambda}})>0.
\]

If $\mathsf m_0$ is obtained from cut-offs and
\[
\liminf_{K\to\infty}Z_K(\beta,\lambda)>0,
\]
then
\[
Z(\beta,\lambda)
=
\lim_{K\to\infty}Z_K(\beta,\lambda)>0.
\]
Thus
\[
0<Z(\beta,\lambda)<\infty.
\]

Equivalently, strict positivity and finiteness may be included directly
in the definition of the admissible domain
\[
\mathcal D
=
\left\{
(\beta,\lambda)
\;\middle|\;
F_{\beta,\lambda}\in\mathcal E(\mathsf m_0)
\right\}.
\]

\subsection{Proof of Theorem~\ref{thm:existence}}

Fix $(\beta,\lambda)\in\mathcal D$ and set
\[
F_{\beta,\lambda}
=
-\beta H-\sum_{a=1}^{N}\lambda^a\Phi_a.
\]
Then
\[
Z(\beta,\lambda)
=
\mathsf m_0(e^{F_{\beta,\lambda}})
\]
satisfies
\[
0<Z(\beta,\lambda)<\infty.
\]

The variational functional is
\[
\mathcal G_{\beta,\lambda}(\mathsf n)
=
\mathcal H(\mathsf n\,\|\,\mathsf m_0)
-
\mathsf n(F_{\beta,\lambda}),
\]
since
\[
-\mathsf n(F_{\beta,\lambda})
=
\beta\,\mathsf n(H)
+
\sum_{a=1}^{N}\lambda^a\mathsf n(\Phi_a).
\]

\medskip
\noindent
\textit{Step 1: Construction of the candidate.}
Define
\[
\mathsf n_{\beta,\lambda}^*(f)
=
\frac{
\mathsf m_0(fe^{F_{\beta,\lambda}})
}{
\mathsf m_0(e^{F_{\beta,\lambda}})
},
\qquad
f\in\mathcal L.
\]
Linearity and positivity follow from those of $\mathsf m_0$. Moreover,
\[
\mathsf n_{\beta,\lambda}^*(1)=1.
\]
Thus
\[
\mathsf n_{\beta,\lambda}^*\in\mathsf S(\mathcal L).
\]

\medskip
\noindent
\textit{Step 2: Fundamental inequality.}
By the definition of relative entropy, for every
$\mathsf n\in\mathsf S(\mathcal L)$,
\[
\mathcal H(\mathsf n\,\|\,\mathsf m_0)
\geq
\mathsf n(F_{\beta,\lambda})
-
\log\mathsf m_0(e^{F_{\beta,\lambda}}).
\]
Hence
\[
\mathcal G_{\beta,\lambda}(\mathsf n)
\geq
-\log Z(\beta,\lambda).
\]

\medskip
\noindent
\textit{Step 3: Equality and uniqueness.}
Proposition~\ref{prop:change-reference}, applied to
\[
F=F_{\beta,\lambda},
\qquad
\mathsf m_F=\mathsf n_{\beta,\lambda}^*,
\]
gives
\[
\mathcal H(
\mathsf n\,\|\,\mathsf n_{\beta,\lambda}^*
)
=
\mathcal G_{\beta,\lambda}(\mathsf n)
+
\log Z(\beta,\lambda).
\]
Taking
\[
\mathsf n=\mathsf n_{\beta,\lambda}^*
\]
and using
\[
\mathcal H(
\mathsf n_{\beta,\lambda}^*
\,\|\,
\mathsf n_{\beta,\lambda}^*
)=0,
\]
we obtain
\[
\mathcal G_{\beta,\lambda}
(\mathsf n_{\beta,\lambda}^*)
=
-\log Z(\beta,\lambda).
\]
Thus $\mathsf n_{\beta,\lambda}^*$ is a minimizer.

If $\mathsf n$ is any other minimizer, then
\[
\mathcal H(
\mathsf n\,\|\,\mathsf n_{\beta,\lambda}^*
)=0.
\]
By Proposition~\ref{prop:entropy-properties},
\[
\mathsf n=\mathsf n_{\beta,\lambda}^*.
\]
Therefore the minimizer is unique.

\medskip
\noindent
\textit{Step 4: Smoothness and covariance Hessian.}
Let
\[
\psi(\beta,\lambda)
=
\log Z(\beta,\lambda).
\]
For a direction
\[
(\dot\beta,\dot\lambda)\in
\mathbb R\times\mathbb R^N,
\]
define
\[
\mathcal O_{\dot\beta,\dot\lambda}
=
\dot\beta H+
\sum_{a=1}^{N}\dot\lambda^a\Phi_a.
\]
Local exponential regularity and
Proposition~\ref{prop:partition-smoothness} give
\[
D\psi(\beta,\lambda)
[\dot\beta,\dot\lambda]
=
-
\mathsf n_{\beta,\lambda}^*
\left(
\mathcal O_{\dot\beta,\dot\lambda}
\right).
\]
Equivalently,
\[
\partial_\beta\psi(\beta,\lambda)
=
-\mathsf n_{\beta,\lambda}^*(H),
\qquad
\partial_{\lambda^a}\psi(\beta,\lambda)
=
-\mathsf n_{\beta,\lambda}^*(\Phi_a).
\]

For two independent directions
\[
(\dot\beta_1,\dot\lambda_1),
\qquad
(\dot\beta_2,\dot\lambda_2),
\]
one obtains
\[
D^2\psi(\beta,\lambda)
\big[
(\dot\beta_1,\dot\lambda_1),
(\dot\beta_2,\dot\lambda_2)
\big]
=
\operatorname{Cov}_{\mathsf n_{\beta,\lambda}^*}
\left(
\mathcal O_{\dot\beta_1,\dot\lambda_1},
\mathcal O_{\dot\beta_2,\dot\lambda_2}
\right).
\]
In particular,
\[
D^2\psi(\beta,\lambda)
[
(\dot\beta,\dot\lambda),
(\dot\beta,\dot\lambda)
]
=
\operatorname{Var}_{\mathsf n_{\beta,\lambda}^*}
\left(
\mathcal O_{\dot\beta,\dot\lambda}
\right)
\geq 0.
\]
Hence $\psi$ is convex. Under the covariance non-degeneracy hypothesis,
its Hessian is positive definite after quotienting
\[
T_{(\beta,\lambda)}\mathcal D
\]
by its thermodynamically null subspace.

Higher derivatives are the corresponding joint cumulants of
\[
H,\Phi_1,\ldots,\Phi_N.
\]

\medskip
\noindent
\textit{Step 5: Legendre--Fenchel duality.}
Define
\[
E(\beta,\lambda)
=
\mathsf n_{\beta,\lambda}^*(H),
\qquad
c_a(\beta,\lambda)
=
\mathsf n_{\beta,\lambda}^*(\Phi_a).
\]
Then
\[
\nabla\psi(\beta,\lambda)
=
\left(
-E(\beta,\lambda),
-c_1(\beta,\lambda),
\ldots,
-c_N(\beta,\lambda)
\right).
\]

The thermodynamic entropy is
\[
\mathcal S(E,c)
=
\inf_{(\beta,\lambda)\in\mathcal D}
\left\{
\psi(\beta,\lambda)
+
\beta E
+
\sum_{a=1}^{N}\lambda^ac_a
\right\}.
\]
Since $\mathcal S$ is an infimum of affine functions of $(E,c)$, it is
concave on its effective domain.

If
\[
(E,c)
=
\left(
E(\beta,\lambda),
c(\beta,\lambda)
\right),
\]
then convexity of $\psi$ shows that $(\beta,\lambda)$ realizes the
infimum and
\[
\mathcal S(E,c)
=
\psi(\beta,\lambda)
+
\beta E
+
\sum_{a=1}^{N}\lambda^ac_a.
\]
Under strict convexity modulo thermodynamically null directions, the
gradient map is locally injective after quotienting each tangent space
by the kernel of the covariance form. A global quotient statement
requires the null spaces to define a fixed or regular integrable
distribution.

This completes the proof of Theorem~\ref{thm:existence}.

\subsection{Proof of Theorem~\ref{thm:geom-eq}}

Let
\[
\psi(\beta,\lambda)
=
\log
\mathsf m_0
\left(
e^{-\beta H-\langle\lambda,\mathbf J\rangle}
\right).
\]
Its convexity follows from the covariance formula, and its smoothness
follows from local exponential regularity.

Assume that $G$ acts on $\mathcal M$ by symplectomorphisms, that
$\mathsf m_0$ is $G$-invariant, that $H$ is $G$-invariant, and that the
moment map is equivariant:
\[
\mathbf J(g\cdot m)
=
\operatorname{Ad}_g^*\mathbf J(m).
\]
We use the convention
\[
\langle\operatorname{Ad}_g^*\mu,\xi\rangle
=
\langle\mu,\operatorname{Ad}_{g^{-1}}\xi\rangle,
\qquad
\mu\in\mathfrak g^*,
\quad
\xi\in\mathfrak g.
\]
For observables, write
\[
(g\cdot f)(m)
=
f(g^{-1}\cdot m).
\]
Then
\[
\mathsf n_{\beta,\lambda}^*(g\cdot f)
=
\mathsf n_{\beta,\operatorname{Ad}_{g^{-1}}\lambda}^*(f).
\]
Therefore the family of equilibrium means is equivariant in the
parameter $\lambda$.

In particular, if $\lambda$ is fixed by the adjoint action, or more
generally if one restricts to the stabilizer
\[
G_\lambda
=
\left\{
g\in G
\;\middle|\;
\operatorname{Ad}_{g^{-1}}\lambda=\lambda
\right\},
\]
then
\[
\mathsf n_{\beta,\lambda}^*(g\cdot f)
=
\mathsf n_{\beta,\lambda}^*(f).
\]
Thus the equilibrium mean is invariant under $G_\lambda$.

The Legendre correspondence between the intensive variables
$(\beta,\lambda)$ and the extensive variables
\[
(E,c)
=
\left(
\mathsf n_{\beta,\lambda}^*(H),
\mathsf n_{\beta,\lambda}^*(\mathbf J)
\right)
\]
is the correspondence generated by the convex potential $\psi$ and its
concave Legendre--Fenchel dual described in
Proposition~\ref{prop:legendre}. This proves
Theorem~\ref{thm:geom-eq}.

\subsection{Proof of Theorem~\ref{thm:stationarity}}

Assume that the Hamiltonian flow $\Phi_t^H$ is globally defined and
that
\[
\mathsf m_0(f\circ\Phi_t^H)
=
\mathsf m_0(f)
\]
for all admissible observables $f$ and all $t\in\mathbb R$. Assume also
that $\mathcal L$ is preserved by the flow. Since
\[
H\circ\Phi_t^H=H
\]
and
\[
\Phi_a\circ\Phi_t^H=\Phi_a,
\qquad
a=1,\ldots,N,
\]
we have
\[
F_{\beta,\lambda}\circ\Phi_t^H
=
F_{\beta,\lambda}.
\]
Therefore
\begin{align*}
\mathsf n_{\beta,\lambda}^*
(f\circ\Phi_t^H)
&=
\frac{
\mathsf m_0\left(
(f\circ\Phi_t^H)e^{F_{\beta,\lambda}}
\right)
}{
Z(\beta,\lambda)
}
\\
&=
\frac{
\mathsf m_0\left(
(fe^{F_{\beta,\lambda}})\circ\Phi_t^H
\right)
}{
Z(\beta,\lambda)
}
\\
&=
\frac{
\mathsf m_0(fe^{F_{\beta,\lambda}})
}{
Z(\beta,\lambda)
}
\\
&=
\mathsf n_{\beta,\lambda}^*(f).
\end{align*}
This proves stationarity.

\subsection{Proof of Theorem~\ref{thm:poisson-kms}}

Let $\mathcal P\subset\mathcal L$ be the Poisson algebra introduced in
Subsection~\ref{subsec:poisson-kms}. Assume that $\mathsf m_0$ is
Poisson invariant:
\[
\mathsf m_0(\{f,g\})=0
\]
for all admissible $f,g\in\mathcal P$.

Let
\[
F_{\beta,\lambda}
=
-\beta H-\sum_{a=1}^{N}\lambda^a\Phi_a.
\]
For admissible $f,g\in\mathcal P$, Poisson invariance gives
\[
\mathsf m_0
\left(
\{f,ge^{F_{\beta,\lambda}}\}
\right)
=
0.
\]
Using the Leibniz rule,
\[
\{f,ge^{F_{\beta,\lambda}}\}
=
\{f,g\}e^{F_{\beta,\lambda}}
+
g\{f,e^{F_{\beta,\lambda}}\}.
\]
The chain rule gives
\[
\{f,e^{F_{\beta,\lambda}}\}
=
e^{F_{\beta,\lambda}}
\{f,F_{\beta,\lambda}\}.
\]
Since
\[
\{f,F_{\beta,\lambda}\}
=
-\beta\{f,H\}
-
\sum_{a=1}^{N}\lambda^a\{f,\Phi_a\},
\]
we obtain
\begin{align*}
\mathsf m_0
\left(
\{f,g\}e^{F_{\beta,\lambda}}
\right)
&=
\beta\,
\mathsf m_0
\left(
g\{f,H\}e^{F_{\beta,\lambda}}
\right)
\\
&\quad+
\sum_{a=1}^{N}\lambda^a
\mathsf m_0
\left(
g\{f,\Phi_a\}e^{F_{\beta,\lambda}}
\right).
\end{align*}
Dividing by $Z(\beta,\lambda)$ gives
\[
\mathsf n_{\beta,\lambda}^*(\{f,g\})
=
\beta\,
\mathsf n_{\beta,\lambda}^*
\bigl(g\{f,H\}\bigr)
+
\sum_{a=1}^{N}\lambda^a
\mathsf n_{\beta,\lambda}^*
\bigl(g\{f,\Phi_a\}\bigr).
\]
This proves Theorem~\ref{thm:poisson-kms}.

```latex
\section*{Appendix: Gaussian and Fourier computations for the partition functional}

The purpose of this appendix is to illustrate, in explicit Fourier
models, the abstract partition functional
\[
Z(\beta,\lambda)
=
\mathsf m_0\left(
e^{-\beta H-\sum_a\lambda^a\Phi_a}
\right).
\]
All formulas below are first understood at the finite-dimensional
cut-off level. Infinite-dimensional expressions are interpreted either
relatively, by subtracting the value at the reference parameters, or
through an explicitly specified Fredholm determinant regularization.

\subsection*{A. 1D EPDiff / Camassa--Holm on \(\mathbb S^1\)}

\paragraph{Setup.}
Let $M=\mathbb S^1$ have length $2\pi$, and expand $u$ in a real
orthonormal Fourier basis. This avoids the double-counting issue
associated with complex coefficients and the constraint
\[
\widehat u_{-k}=\overline{\widehat u_k}.
\]
The notation $|k|\leq N$ below refers to the constant mode together with
the sine and cosine modes of frequencies $1,\ldots,N$, so that the
cut-off space has dimension $2N+1$.

For the Camassa--Holm $H^1$ metric,
\[
A=1-\alpha^2\partial_x^2,
\qquad
A_k=1+\alpha^2k^2,
\]
and
\[
H(u)
=
\frac12\sum_k A_k u_k^2.
\]
Let the reference Gaussian be centered with covariance
\[
C_k=A_k^{-1}.
\]
Equivalently, at the cut-off level, the modes $u_k$ are independent
Gaussian variables with variance $A_k^{-1}$.

The translation momentum is
\[
\Phi(u)
=
\int_{\mathbb S^1}u(x)\,\mathrm dx
=
2\pi u_0.
\]
Set
\[
\theta=2\pi\lambda.
\]
Then
\[
-\lambda\Phi(u)=-\theta u_0.
\]

\paragraph{Cut-off partition functional.}
For the modes of frequencies at most $N$, define
\[
Z_N(\beta,\lambda)
=
\mathbb E
\left[
\exp\left(
-\frac{\beta}{2}\sum_{|k|\leq N}A_k u_k^2
-\theta u_0
\right)
\right].
\]
Since the modes are independent,
\[
\log Z_N(\beta,\lambda)
=
-\frac12\sum_{|k|\leq N}\log(1+\beta)
+
\frac{\theta^2}{2(1+\beta)}.
\]
Thus
\[
\boxed{
\log Z_N(\beta,\lambda)
=
-\frac{2N+1}{2}\log(1+\beta)
+
\frac{(2\pi\lambda)^2}{2(1+\beta)}.
}
\]
The first term diverges as $N\to\infty$. The partition functional
relative to the value at $\lambda=0$ is therefore
\[
\log Z_{\mathrm{rel}}(\beta,\lambda)
=
\lim_{N\to\infty}
\left(
\log Z_N(\beta,\lambda)-\log Z_N(\beta,0)
\right)
=
\frac{(2\pi\lambda)^2}{2(1+\beta)}.
\]

\paragraph{Normalizability.}
The modified precision is $(1+\beta)A_k$, hence the Gaussian integral
is finite whenever
\[
1+\beta>0.
\]
In particular, the thermodynamic range $\beta>0$ is admissible. For a
general linear constraint
\[
\Phi_\ell(u)=\sum_k\ell_k u_k,
\]
one obtains
\[
\log Z_{\mathrm{rel}}(\beta,\ell)
=
\frac12
\sum_k
\frac{C_k|\ell_k|^2}{1+\beta},
\]
provided
\[
\sum_k C_k|\ell_k|^2<\infty.
\]
This is the finite-variance, or Cameron--Martin dual, admissibility
condition for the linear observable.

\subsection*{B. 2D Euler vorticity on \(\mathbb T^2\)}

\paragraph{Setup.}
Let
\[
\omega(x)
=
\sum_{k\in\mathbb Z^2\setminus\{0\}}
\omega_k e^{ik\cdot x}
\]
be a real mean-zero vorticity field, with
\[
\omega_{-k}=\overline{\omega_k}.
\]
In all products and sums below, the modes are counted in a real
orthonormal Fourier basis, or equivalently over an independent
half-lattice with the appropriate multiplicities. In particular, no
mode is counted twice.

The kinetic energy is
\[
H(\omega)
=
\frac12
\sum_{k\neq 0}
\frac{|\omega_k|^2}{|k|^2}.
\]
Let the reference Gaussian have modal covariance
\[
Q_k=\frac{1}{|k|^2}.
\]
At finite cut-off level, the independent real Fourier coordinates are
Gaussian with variance $Q_k$.

In two dimensions,
\[
\sum_{k\neq 0}Q_k
=
\sum_{k\neq 0}\frac{1}{|k|^2}
=
\infty.
\]
Consequently, this covariance does not define a Gaussian Radon measure
on the $L^2$ vorticity space. It may instead be realized on a suitable
negative Sobolev space, or treated through the finite-dimensional
cut-offs used below.

We include a quadratic enstrophy constraint
\[
C_2(\omega)
=
\frac12\sum_{k\neq 0}|\omega_k|^2
\]
with multiplier $\alpha$. The cut-off partition is
\[
Z_N(\beta,\alpha)
=
\mathbb E
\left[
\exp\left(
-\frac{\beta}{2}
\sum_{0<|k|\leq N}
\frac{|\omega_k|^2}{|k|^2}
-
\frac{\alpha}{2}
\sum_{0<|k|\leq N}
|\omega_k|^2
\right)
\right].
\]

\paragraph{Mode factorization.}
For a centered real Gaussian variable of variance $Q_k$,
\[
\mathbb E
\left(
e^{-\frac12t_k|\omega_k|^2}
\right)
=
(1+Q_kt_k)^{-1/2},
\]
where
\[
t_k=\frac{\beta}{|k|^2}+\alpha.
\]
Since $Q_k=|k|^{-2}$, one gets
\[
1+Q_kt_k
=
1+\frac{\beta}{|k|^4}
+\frac{\alpha}{|k|^2}.
\]
Hence, with each independent real mode counted once,
\[
\boxed{
\log Z_N(\beta,\alpha)
=
-\frac12
\sum_{0<|k|\leq N}
\log\left(
1+\frac{\beta}{|k|^4}
+\frac{\alpha}{|k|^2}
\right).
}
\]

\paragraph{Admissible domain.}
The finite-dimensional Gaussian integrals are finite precisely when
\[
1+\frac{\beta}{|k|^4}+\frac{\alpha}{|k|^2}>0
\qquad
\text{for all }0<|k|\leq N.
\]
A sufficient positivity condition, uniform in $N$, is
\[
\beta>0,
\qquad
\alpha\geq 0.
\]
More generally, one works on the domain where all modified precisions
remain positive.

For $\alpha\neq 0$, the series
\[
\sum_{k\neq 0}
\frac{\alpha}{|k|^2}
\]
is not summable in two dimensions. Thus positivity of the modal
precisions does not by itself imply convergence to a nonzero
infinite-dimensional partition functional. A relative normalization or
a $\det_2$-type renormalization is required for the enstrophy tilt.

\paragraph{Nonlinear Casimirs.}
For nonlinear Casimirs
\[
C_\phi(\omega)
=
\int_{\mathbb T^2}\phi(\omega)\,\mathrm dx,
\]
the partition functional is no longer mode-diagonal unless $\phi$ is
quadratic. It should therefore be treated either perturbatively, through
finite-dimensional cut-offs, or within the general normalized-mean
variational formalism. Quadratic Casimirs are the analytically closed
case because they preserve Gaussian factorization.

\subsection*{C. Remarks on renormalization and determinants}

The divergent constants appearing in infinite products may be treated
by one of the following procedures:
\begin{itemize}
\item \emph{relative normalization}, in which one considers a difference
\[
\log Z_{\mathrm{rel}}(\theta,\theta_0)
=
\log Z(\theta)-\log Z(\theta_0);
\]
\item \emph{Fredholm determinant regularization}, when the perturbation
of the covariance is trace class;
\item \emph{Carleman--Fredholm $\det_2$-regularization}, when the
perturbation is Hilbert--Schmidt but not trace class.
\end{itemize}
These procedures are not automatically equivalent: their agreement
depends on the chosen reference parameter and on the corresponding
counterterms. The renormalization convention must therefore be kept
fixed throughout a given model.

Only the resulting relative or renormalized free energy and its
derivatives enter the corresponding renormalized variational principle.

\subsection*{D. Summary of normalizability in Fourier--Gaussian examples}

\begin{itemize}
\item \textbf{Energy tilt.}
If the reference covariance is the inverse of the quadratic inertia,
then the energy tilt modifies each precision by the factor $1+\beta$.
Hence
\[
1+\beta>0
\]
is sufficient at every finite cut-off, and in particular all
$\beta>0$ are admissible at that level. Existence of an
infinite-dimensional partition functional may still require relative
normalization.

\item \textbf{Linear constraints.}
For a linear observable
\[
\Phi_\ell(u)=\sum_k\ell_k u_k,
\]
finiteness requires
\[
\sum_k C_k|\ell_k|^2<\infty.
\]
This is the finite-variance condition for the Gaussian linear tilt.

\item \textbf{Quadratic constraints.}
For quadratic constraints, positivity of the modified precision
operator is necessary for Gaussian integrability. In infinite
dimension, convergence of the unrenormalized determinant additionally
requires the corresponding covariance perturbation to be trace class.
If it is only Hilbert--Schmidt, a $\det_2$-renormalization may be used.
\end{itemize}

\subsection*{E. Gaussian partition with quadratic and linear tilt}

\begin{Proposition}[Gaussian quadratic-linear partition]
\label{prop:gaussian-quadratic-linear}
Let $\mathcal H$ be a real separable Hilbert space and let
\[
\mu_0=\mathcal N(0,C)
\]
be a centered Gaussian Radon measure on $\mathcal H$, where $C$ is a
positive, selfadjoint, injective trace-class operator. Let $T$ be a
bounded selfadjoint operator and set
\[
K=C^{1/2}TC^{1/2}.
\]
Assume that
\[
I+K\geq\delta I
\]
for some $\delta>0$. Then $K$ is trace class. For every
$h\in\mathcal H$, the integral
\[
Z
=
\int_{\mathcal H}
\exp\left(
-\frac12\langle u,Tu\rangle+\langle h,u\rangle
\right)
\,\mathrm d\mu_0(u)
\]
is finite and
\[
\boxed{
\log Z
=
-\frac12\log\det(I+K)
+
\frac12
\left\langle
h,\,
C^{1/2}(I+K)^{-1}C^{1/2}h
\right\rangle.
}
\]
The operator
\[
C_T
=
C^{1/2}(I+K)^{-1}C^{1/2}
\]
is the covariance of the tilted Gaussian measure. In the sense of
quadratic forms on the Cameron--Martin space,
\[
C_T=(C^{-1}+T)^{-1}.
\]
\end{Proposition}

\begin{proof}
In finite dimension, completion of the square gives
\[
\begin{aligned}
Z
&=
(\det C)^{-1/2}
\det(C^{-1}+T)^{-1/2}
\exp\left(
\frac12
\langle h,(C^{-1}+T)^{-1}h\rangle
\right)
\\
&=
\det(I+C^{1/2}TC^{1/2})^{-1/2}
\exp\left(
\frac12
\left\langle
h,
C^{1/2}(I+K)^{-1}C^{1/2}h
\right\rangle
\right).
\end{aligned}
\]

In infinite dimension, choose increasing finite-rank orthogonal
projections $P_N\to I$ strongly which commute with $C$, and set
\[
C_N=P_NCP_N,
\qquad
T_N=P_NTP_N,
\qquad
h_N=P_Nh,
\]
and
\[
K_N=C_N^{1/2}T_NC_N^{1/2}.
\]
Since $C$ is trace class and $T$ is bounded,
\[
K_N\longrightarrow K
\]
in trace norm. Consequently,
\[
\det(I+K_N)\longrightarrow\det(I+K).
\]
Moreover, the strict positivity
\[
I+K\geq\delta I
\]
and the operator-norm convergence of $K_N$ imply that
$(I+K_N)^{-1}$ is uniformly bounded for all sufficiently large $N$ and
converges strongly to $(I+K)^{-1}$. Since
\[
C_N^{1/2}\longrightarrow C^{1/2}
\]
strongly and $h_N\to h$, the corresponding quadratic terms converge to
\[
\left\langle
h,
C^{1/2}(I+K)^{-1}C^{1/2}h
\right\rangle.
\]
Passing to the limit in the finite-dimensional formula proves the
result.
\end{proof}

\begin{rem}[Hilbert--Schmidt renormalization]
Suppose more generally that $K$ is Hilbert--Schmidt, that
\[
I+K\geq\delta I
\]
for some $\delta>0$, but that $K$ is not trace class. Then
$\det(I+K)$ is not defined, whereas the Carleman--Fredholm determinant
$\det_2(I+K)$ is well defined. The renormalized logarithmic partition
functional may be defined by
\[
\boxed{
\log Z_{\mathrm{ren}}
=
-\frac12\log\det\nolimits_2(I+K)
+
\frac12
\left\langle
h,
C^{1/2}(I+K)^{-1}C^{1/2}h
\right\rangle.
}
\]
This definition amounts to subtracting the divergent first-order trace
term at the finite-dimensional level. No quantity
$\operatorname{Tr}(K)$ is used when $K$ is not trace class.

If $K$ is trace class, then
\[
\det\nolimits_2(I+K)
=
\det(I+K)e^{-\operatorname{Tr}(K)},
\]
so that
\[
-\frac12\log\det(I+K)
=
-\frac12\log\det\nolimits_2(I+K)
-\frac12\operatorname{Tr}(K).
\]
Thus the ordinary and renormalized formulas differ by the explicit
first-order trace counterterm whenever both are defined.
\end{rem}

\subsection*{F. Gaussian quadratic-linear tilts under normalized means}

We now reinterpret the preceding computation in the normalized-mean
setting. Let $P_N$ be an increasing family of finite-rank orthogonal
projections such that
\[
P_N\to I
\]
strongly and $P_NC=CP_N$. Define
\[
C_N=P_NCP_N,
\qquad
T_N=P_NTP_N,
\qquad
h_N=P_Nh,
\]
and
\[
K_N=C_N^{1/2}T_NC_N^{1/2}.
\]
On $P_N\mathcal H$, let
\[
\mu_{0,N}=\mathcal N(0,C_N),
\]
and set
\[
Z_N(T,h)
=
\int_{P_N\mathcal H}
\exp\left(
-\frac12\langle u,T_Nu\rangle+\langle h_N,u\rangle
\right)
\,\mathrm d\mu_{0,N}(u).
\]
Since $\mu_{0,N}$ is normalized,
\[
Z_N(0,0)=1.
\]
The normalized-mean log-partition functional is defined by
\[
\log Z_{\mathrm{rel}}(T,h)
=
\lim_{N\to\infty}\log Z_N(T,h),
\]
whenever this limit exists and is independent of the admissible
exhaustion.

\begin{Proposition}[Relative Gaussian formula for normalized means]
\label{prop:normalized-mean-gauss}
Assume:
\begin{enumerate}
\item
\[
I+K\geq\delta I
\]
for some $\delta>0$;

\item
\[
K=C^{1/2}TC^{1/2}
\]
is trace class;

\item
\[
K_N\longrightarrow K
\]
in trace norm;

\item
\[
P_Nh\longrightarrow h
\]
in $\mathcal H$.
\end{enumerate}
Then
\[
\boxed{
\log Z_{\mathrm{rel}}(T,h)
=
-\frac12\log\det(I+K)
+
\frac12
\left\langle
h,\,
C^{1/2}(I+K)^{-1}C^{1/2}h
\right\rangle.
}
\]
The value is independent of the exhaustion among admissible families
satisfying the convergence assumptions above.
\end{Proposition}

\begin{proof}
For each $N$, the finite-dimensional formula gives
\[
\log Z_N(T,h)
=
-\frac12\log\det(I+K_N)
+
\frac12
\left\langle
h_N,\,
C_N^{1/2}(I+K_N)^{-1}C_N^{1/2}h_N
\right\rangle.
\]
Trace-norm convergence
\[
K_N\longrightarrow K
\]
gives
\[
\det(I+K_N)\longrightarrow\det(I+K).
\]
It also gives operator-norm convergence. Hence, using
\[
I+K\geq\delta I,
\]
the inverses $(I+K_N)^{-1}$ are uniformly bounded for all sufficiently
large $N$ and converge strongly to $(I+K)^{-1}$.

Since
\[
C_N^{1/2}\longrightarrow C^{1/2}
\]
strongly and $h_N\to h$, the quadratic terms converge to
\[
\left\langle
h,\,
C^{1/2}(I+K)^{-1}C^{1/2}h
\right\rangle.
\]
The limit is therefore the stated expression.
\end{proof}

\begin{rem}[Connection with the Donsker--Varadhan entropy]
For
\[
f(u)
=
-\frac12\langle u,Tu\rangle+\langle h,u\rangle,
\]
the previous proposition computes
\[
\Lambda_{\mathsf m_0}(f)
=
\log\mathsf m_0(e^f)
\]
through finite-dimensional normalized Gaussian cut-offs. Hence the
relative entropy
\[
\mathcal H(\mathsf n\,\|\,\mathsf m_0)
=
\sup_f
\left\{
\mathsf n(f)-\Lambda_{\mathsf m_0}(f)
\right\}
\]
recovers the same quadratic-linear exponential equilibria as in the
main text, whenever $f$ belongs to the corresponding exponential
domain.
\end{rem}

\vskip 12pt
```

\paragraph{\bf Data availability statement} No data is available for this work.

\vskip 12pt

\paragraph{\bf Conflict of interest statement} The author declares no conflict of interest.

\vskip 12pt

\paragraph{\bf Funding} No funding supported this work.

\vskip 12pt

\paragraph{\bf Acknowledgements} J.-P.M thanks the France 2030 framework programme Centre Henri Lebesgue ANR-11-LABX-0020-01 
for creating an attractive mathematical environment.

\vskip 12pt

\paragraph{\bf Author's Note on AI Assistance}
Portions of the text were developed with the assistance of a generative language model (OpenAI ChatGPT, based on the GPT-4 architecture). The AI was used to assist with drafting, editing, and standardizing the bibliography format. All mathematical content, structure, and theoretical constructions were provided, verified, and curated by the author. The author assumes full responsibility for the correctness, originality, and scholarly integrity of the final manuscript.

\end{document}